\documentclass[journal,twoside,web]{ieeecolor}
\usepackage[dvipsnames]{xcolor}

\usepackage{generic}
\usepackage{cite}
\usepackage{amsmath,amssymb,amsfonts}
\usepackage{algorithmic}
\usepackage{graphicx}
\usepackage[OT2,T1]{fontenc}
\usepackage{MnSymbol}
\usepackage{lipsum}
\usepackage{color}
\newtheorem{theorem}{Theorem}

\newtheorem{definition}{Definition}
\newtheorem{property}{Property}

\newtheorem{remark}{Remark}
\DeclareSymbolFont{cyrletters}{OT2}{wncyr}{m}{n}
\DeclareMathSymbol{\Sha}{\mathalpha}{cyrletters}{"58}

\newcommand{\T}{\mathcal{T}}

\usepackage{textcomp}
\def\BibTeX{{\rm B\kern-.05em{\sc i\kern-.025em b}\kern-.08em
 T\kern-.1667em\lower.7ex\hbox{E}\kern-.125emX}}
\begin{document}

\title{A TBLMI Framework for Harmonic Robust Control}
\author{Flora Vernerey, Pierre Riedinger and Jamal Daafouz\\% <-this % stops a space
%\thanks{This work is supported by HANDY project ANR-18-CE40-0010-02.}% <-this % stops a space
\thanks{The authors are with Universit\'e de Lorraine, CNRS, CRAN, F-54000 Nancy, France. (email: flora.vernerey@univ-lorraine.fr; pierre.riedinger@univ-lorraine.fr; jamal.daafouz@univ-lorraine.fr)}}
\maketitle
%%%%%%%%%%%%%%%%%%%%%%%%%%%%%%%%%%%%%%%%%%%%%%%%%%%%%%%%%%%%%%%%%%%%%%%%%%%%%%%%
\begin{abstract}
%In the context of periodic systems, the fundamental principle of harmonic control revolves around the transformation of the control design problem into the harmonic domain, where the system model becomes time-invariant. This crucial shift significantly simplifies the design process, rendering it amenable to any time-invariant method. However, a major challenge arises due to the fact that this time-invariant model is inherently of infinite dimension. 
%Building upon the well-established success of Linear Matrix Inequalities (LMIs) in the realm of robust control, 
The primary objective of this paper is to
%to extend this paradigm to harmonic robust control. Our work 
demonstrate that problems related to stability and robust control in the harmonic context can be effectively addressed by formulating them as semidefinite optimization problems, invoking the concept of infinite-dimensional Toeplitz Block LMIs (TBLMIs). One of the central challenges tackled in this study pertains to the efficient resolution of these infinite-dimensional TBLMIs. Exploiting the structured nature of such problems, we introduce a consistent truncation method that effectively reduces the problem to a finite-dimensional convex optimization problem. By consistent we mean that the solution to this finite-dimensional problem allows to closely approximate the infinite-dimensional solution with arbitrary precision. Furthermore, we establish a link between the harmonic framework and the time domain setting, emphasizing the advantages over Periodic Differential LMIs (PDLMIs). We illustrate that our proposed framework is not only theoretically sound but also practically applicable to solving  $H_2$ and $H_\infty$ harmonic control design problems. To enable this, we extend the definitions of $H_2$ and $H_\infty$ norms into the harmonic space, leveraging the concepts of the harmonic transfer function and the average trace operator for Toeplitz Block operators. Throughout this paper, we support our theoretical contributions with a range of illustrative examples that demonstrate the effectiveness of our approach.
\end{abstract}
%\begin{IEEEkeywords}
%, Dynamic phasors, Harmonic modeling and control
%\end{IEEEkeywords}

%%%%%%%%%%%%%%%%%%%%%%%%%%%%%%%%%%%%%%%%%%%%%%%%%%%%%%%%%%%%%%%%%%%%%%%%%%%%%%%%
\section{Introduction}
Harmonics, as sinusoidal components of a signal that occur at integer multiples of a fundamental frequency, are omnipresent in a wide range of applications, including power systems, communications, mechanical systems, and electronic devices. These harmonics can be introduced by various sources, such as nonlinearities in the system dynamics, external interference, or the intrinsic nature of the system itself. Their presence can lead to undesirable effects, such as oscillations, instability, and increased energy consumption, making the management of harmonics a pressing concern in control engineering. A successful way of tracking or rejecting periodic signals is repetitive control, a well-known technique based on the internal model principle \cite{Hara88,Lee98,WEISS99}. Output regulation methods that may resort to a harmonic representation are also of interest in this context \cite{Astolfi15,Astolfiscl22}.   
The prominence of harmonic control arises from its ability to capture and manipulate the often-pervasive harmonic disturbances present in practical systems. Its significance lies in the fact that it enables the transformation of control design problems into time-invariant ones, albeit in an infinite-dimensional framework \cite{sanders1991generalized,wereley1990analysis,zhou_stability_2001,zhou2004spectral}. This transformation allows for the utilization of well-established time-invariant control design techniques. Consequently, the control of systems in the presence of harmonics becomes more tractable and intuitive. However, the transition to infinite dimensionality presents its own set of challenges that are not encountered in finite-dimensional control \cite{riedinger2022harmonic}. The development of innovative tools and methodologies is essential to effectively harness the potential of harmonic control and address the intricacies introduced by the infinite-dimensional modeling framework.

\vspace{.1cm}
Lyapunov and Riccati equations have proven to be a powerful and versatile tool for addressing control problems in a variety of settings. They have offered practical solutions to finite-dimensional systems. The key point in utilizing such methods in the realm of harmonic control design lies in their successful extension to address and approximate infinite-dimensional problems. Recently, in \cite{riedinger2022solving}, this extension has been realized, emphasizing the critical role of practical solution determination with minimal errors in ensuring the stabilization properties.
Nonetheless, the demand for a more adaptable framework to tackle robust control challenges within the harmonic context is evident. Building upon the success of LMIs in robust control, the development of tools for solving  infinite-dimensional
Toeplitz Block LMIs (TBLMIs) stands as a matter of paramount significance in advancing robust control strategies tailored for the domain of harmonic systems.

\vspace{.1cm}
In this paper, our central focus is on solving TBLMIs and their practical application for robust $H_2$ and $H_\infty$ harmonic control design. Our results are a noteworthy and valuable alternative to prior research works \cite{wereley1990analysis,zhou2002h2}. In \cite{wereley1990analysis}, Floquet theory is used to derive an infinite-dimensional harmonic state space representation from which a harmonic transfer function is deduced. Such a frequency domain representation allows to define operators for $H_2$ and $H_\infty$ robust control but, as the related frequency response operators are infinite dimensional, the numerical implementation is nontrivial \cite{zhou2002h2}. Also, our proposed methodology complements traditional time-domain methods for periodic systems \cite{colaneri2000continuous}. These conventional approaches often involve solving control design problems using differential Riccati equations and/or inequalities \cite{colaneri2000continuous2,Gero23}. Alternatively, they may employ lifting schemes  \cite{bamieh1991lifting,bamieh1992general,yamamoto1996frequency,khargonekar1991h2} to address robust control problems. By exploring the possibilities inherent in infinite-dimensional TBLMIs, we address the related difficulties and offer a new solution in the domain of harmonic-based robust control. Our central aim is to devise a truncated version of the original infinite dimensional problem, allowing for the retrieval of an infinite-dimensional solution with any desired degree of accuracy. It is important to note that our approach does not rely on the conventional application of Floquet theory \cite{wereley1990analysis, zhou2004spectral,zhou2008derivation}, which, while valuable for stability analysis in Linear Time-Periodic (LTP) systems, exhibits limitations when employed for control design purposes~\cite{riedinger2022solving}.  Compared to the preliminary version in \cite{floracdc23}, the results we propose in this paper are more general and do not rely on restrictive assumptions. We also provide new results based on precise definitions of $H_2$ and $H_\infty$ norms in the harmonic domain, illuminating their implications in the time domain. In particular, defining the harmonic $H_2$ norm
is not a straightforward task in the harmonic domain, primarily because the Frobenius
norm is not trace-class. The definitions we propose are underpinned by the introduction of harmonic transfer functions, trace operators tailored for infinite-dimensional Toeplitz blocks, and operators bounded on $\ell^2$, each endowed with relevant mathematical properties. Additionally, we underscore the remarkable equivalence between solving a TBLMI in the harmonic domain and addressing a periodic differential LMI (PDLMI) in the time domain. This observation underscores the simplification achieved by adopting a harmonic formulation. 

\vspace{.1cm}
The paper is organized as follows. The next section is dedicated to mathematical preliminaries and harmonic modeling. In Section III, a framework for harmonic robust control is provided. This framework includes the definition of harmonic transition and transfer functions, the definition of trace operator for Toeplitz block operators bounded on $\ell^2$, the definition of $H_2$ and $H_\infty$ norms in the harmonic domain as well as their meaning in the time domain. Moreover, we formulate harmonic robust control problems in terms of infinite dimensional convex optimization problems involving TBLMI constraints. Section IV is devoted to infinite dimensional TBLMIs  and their equivalence with PDLMIs in the time domain. Section V is dedicated to a consistent truncation procedure which allows to recover, up to an arbitrarily small error, the infinite-dimensional solution to a TBLMI based convex optimization problem by solving a finite dimensional truncated problem. Finally, before concluding, we illustrate the results of this paper in section VI and apply the proposed procedure to design harmonic $H_2$ and $H_\infty$ optimal state feedbacks.

\vspace{.1cm}
{\bf Notations: } The transpose of a matrix $A$ is denoted $A'$ and $A^*$ denotes the complex conjugate transpose $A^*=\bar A'$. $\textsf{j}$ is the imaginary unit. The $n$-dimensional identity matrix is denoted $Id_n$. The infinite identity matrix is denoted $\mathcal{I}$. For $m\in\mathbb{Z}^+\cup \{\infty\}$, the flip matrix $J_m$ is the $(2m+1) \times (2m+1)$ matrix having 1 on the anti-diagonal and zeros elsewhere. 
$C^a$ denotes the space of absolutely continuous function,
$L^{p}([a\ b],\mathbb{C}^n)$ (resp. $\ell^p(\mathbb{C}^n)$) denotes the Lebesgues spaces of $p-$integrable functions on $[a, b]$ with values in $\mathbb{C}^n$ (resp. $p-$summable sequences of $\mathbb{C}^n$) for $1\leq p\leq\infty$. $L_{loc}^{p}$ is the set of locally $p-$integrable functions. The notation $f(t)=g(t)\ a.e.$ means almost everywhere in $t$ or for almost every $t$. 
To simplify the notations, $L^p([a,b])$ or $L^p$ will be often used instead of $L^p([a,b],\mathbb{C}^n)$. Finally, $\otimes$ denotes the Kroenecker product.
%For example, $x\in L^2([a,b])$ means $x \in L^2([a,b],\mathbb{C}^n)$. %We denote by $col(X)$ the vectorization of a matrix $X$, formed by stacking the columns of $X$ into a single column vector. Finally, $<\cdot,\cdot>$ refers to the scalar product in $\ell^2$.
%%%%%%%%%%%%%%%%%%%%%%%%%%%%%%%%%%%%%%%%%%%%%%%%%%%%%%%%%%%%%%%%%%%%%%%%%%%%%%%%
%\vspace{-.3cm}
\section{Preliminaries}

\subsection{Sliding Fourier decomposition and basic results}
Consider $x\in L^{2}_{loc}(\mathbb{R},\mathbb{C})$ a complex valued function of time. Its sliding Fourier decomposition over a window of length $T$ is defined by the time-varying infinite sequence $X:=\mathcal{F}(x)\in C^a(\mathbb{R},\ell^2(\mathbb{C}))$ (see \cite{blin_necessary_2022}) whose components satisfy:
$$X_{k}(t):=\frac{1}{T}\int_{t-T}^t x(\tau)e^{-\textsf{j}\omega k \tau}d\tau$$ for $k\in \mathbb{Z}$, with $\omega:=\frac{2\pi}{T}$.
If $x=(x_1,\cdots,x_n)\in L^{2}_{loc}(\mathbb{R},\mathbb{C}^n)$ is a complex valued vector function, then
$$X:=\mathcal{F}(x)=(\mathcal{F}(x_1), \cdots,\mathcal{F}(x_n)).$$
The time varying vector $X_k=(X_{1,k}, \cdots, X_{n,k})$ with $$X_{k}(t):=\frac{1}{T}\int_{t-T}^t x(\tau)e^{-\textsf{j}\omega k \tau}d\tau$$
is called the $k-$th phasor of $X$. Note also that when $x$ is a real signal, $X_k=\overline{ X_{-k}}$ for any $k$.

Additionally, we consider the Toeplitz transformation of a scalar function $x \in L^2_{\text{loc}}$, which is defined as an infinite-dimensional Toeplitz matrix function given by:
\begin{align*}
	\mathcal{T}(x):=
	\left[
	\begin{array}{ccccc}
		\ddots & & \vdots & &\udots \\ & x_{0} & x_{-1} & x_{-2} & \\
		\cdots & x_{1} & x_{0} & x_{-1} & \cdots \\
		& x_{2} & x_{1} & x_{0} & \\
		\udots & & \vdots & & \ddots\end{array}\right],\end{align*}
where the terms $x_k$, $k\in\mathbb{Z}$ refer to the phasors of $X:=\mathcal{F}(x)$.
For a $n \times m$ matrix  function $A=(a_{ij})$ ($\in L^2_{loc}(\mathbb{R},\mathbb{C}^{n\times m}$), we define the associated infinite dimensional Toeplitz Block (TB) matrix function as
\begin{equation} \label{ToepFormMat}
\mathcal{A} :=\mathcal{T}(A)= \begin{bmatrix}
\mathcal{A}_{11} & \ldots  & \mathcal{A}_{1n} \\ 
\vdots&&\vdots \\ 
\mathcal{A}_{m1} &\ldots & \mathcal{A}_{mn}
\end{bmatrix} 
\end{equation}
with $\mathcal{A}_{ij}:=\mathcal{T}(a_{ij})$, $i=1,\cdots, n,\ j=1,\cdots, m$.
We also recall the following essential rules:
\begin{enumerate}
    \item For a matrix function $A \in L_{loc}^\infty$ and a vector $x\in L_{loc}^2$:
\begin{equation}
  \mathcal{F}\left(Ax\right)=\mathcal{T}(A)\mathcal{F}(x)=\mathcal{A}X\label{matvectproduit}
\end{equation}
where $X:=\mathcal{F}(x)$ and  $\mathcal{A}:=\mathcal{T}\left(A\right)$.
\item For two matrix functions $A\in L_{loc}^\infty$ and $B\in L_{loc}^\infty$:
\begin{equation}
  \mathcal{T}\left(AB\right)=\mathcal{T}\left(A\right)\mathcal{T}\left(B\right)=\mathcal{AB}\label{produit}
\end{equation}
\end{enumerate}
According to the sliding Fourier decomposition, all these transformations lead to absolutely continuous functions of time.
When a $T$-periodic vector $x$ or matrix function $A$ are considered, both $X:=\mathcal{F}(x)$ and $\mathcal{A}:=\mathcal{T}(A)$ are constant. %More specifically, if $A$ is a constant matrix, we obtain the following relationship:
%\begin{equation}
%  \mathcal{A}:=\mathcal{T}\left(A\right)= A \otimes\mathcal{I}.
%\end{equation}

Let us recall some fundamental results related to the sliding Fourier decomposition.
\begin{theorem}\label{borne} Consider a $T$-periodic matrix function $A\in L^2([0 \ T],\mathbb{C}^{n\times m})$. Then, $\mathcal{A}:=\mathcal{T}(A)$ is a constant bounded operator on $\ell^2$ if and only if $A\in L^{\infty}([0\ T],\mathbb{C}^{n\times m} )$.
Moreover, the operator norm induced by the $\ell^2$-norm satisfies : $$\|\mathcal{A}\|_{\ell^2}:=\sup_{\|X\|_{\ell^2}=1}\|\mathcal{A}X\|_{\ell^2}=\|A\|_{L^{\infty}}$$
\end{theorem}
\begin{proof}See Part V p.p. 562-574  of \cite{gohberg1993classes}.
\end{proof}

\begin{definition}\label{H} We say that $X$ belongs to $H$ if $X$ is an absolutely continuous function (i.e $X\in C^a(\mathbb{R},\ell^2(\mathbb{C}^n))$ and fulfils, for any $k$, the following condition: \begin{equation*}\dot X_k(t)=\dot X_0(t)e^{- \textsf{j}\omega k t} \ a.e.\end{equation*}
\end{definition}

Similarly to the Riesz-Fisher theorem which establishes a one-to-one correspondence between the spaces $L^2$ and $\ell^2$, the following theorem establishes a one-to-one correspondence between the spaces $L_{loc}^2$ and $H$.
\begin{theorem}\label{coincidence}For a given $X\in L_{loc}^{\infty}(\mathbb{R},\ell^2(\mathbb{C}^n))$, there exists a representative $x\in L^2_{loc}(\mathbb{R},\mathbb{C}^n)$ of $X$, i.e. $X=\mathcal{F}(x)$, if and only if $X \in H$.
\end{theorem}
\begin{proof}See \cite{blin_necessary_2022}.
\end{proof}

For harmonic control design, Theorem~\ref{coincidence} has a fundamental consequence: the design of a control $U$ in the harmonic domain must belong to the space $H$, otherwise its time domain counterpart $u$ does not exist. For example, if one attempts to design a state feedback $U:=-\mathcal{K}X$ with a constant gain operator $\mathcal{K}$, it is proven that $\mathcal{K}$ must be a Toeplitz block operator \cite{blin_necessary_2022}. Hence, $u$ is determined by $u:=-Kx$ where $K=\sum_{k\in\mathbb{Z}} K_ke^{{\rm j}\omega kt}$. 
Moreover, to ensure controlled boundedness in the control signal $u$, it is imperative that $\mathcal{K}$  represents a bounded operator on $\ell^2$. This condition, in accordance with Theorem~\ref{borne}, implies that $K\in L^\infty$. 

For clarity reasons, when considering a lowercase vector-valued function $x \in L^{2}_{\text{loc}}(\mathbb{R},\mathbb{C}^n)$, we adopt the following convention: $X := \mathcal{F}(x)$ represents the Fourier transform of $x$, while $\mathcal{ X} := \mathcal{T} (x)$ represents the Toeplitz form of $x$. 
Through a slight abuse of notation, we may interchangeably use $\T(X)$ instead of $\T(x)$  to signify the function responsible for constructing a Toeplitz matrix from the elements of $X$, where $X$ is defined as $\mathcal{F}(x)$.

 {\remark In this paper, we employ a TB matrix representation rather than the more conventional Block Toeplitz (BT) matrix representation. The primary motivation behind this choice lies in its ability to yield a harmonic equation structure akin to that found in the time domain, as exemplified in \eqref{ToepFormMat}, for instance. This choice better suits our objectives for analysis and control design.
 To obtain a Block Toeplitz (BT) structure as in \cite{blin_necessary_2022, zhou2008derivation}, one has to consider the transformation operator $\mathcal{F}(x)=( \cdots, X_{-1},X_{0}, X_{1},\cdots)$, where $X_k$ refers to the $k-th$ phasors, instead of $\mathcal{F}(x)=(\mathcal{F}(x_1), \cdots,\mathcal{F}(x_n)).$ Obviously, we can readily switch between these two representations by applying an appropriate permutation matrix.}

\subsection{Harmonic modeling}
By leveraging Theorem~\ref{coincidence}, any system with solutions in the Caratheodory sense can be transformed through a sliding Fourier decomposition into an infinite-dimensional system. This transformation establishes a direct correspondence between the trajectories of the original system and those within the infinite-dimensional space, provided that the latter belong to the subspace $H$. Furthermore, in the context of a periodic system with a period $T$, the resulting infinite-dimensional system is time-invariant. For instance, with a mild assumption (as detailed in \cite{blin_necessary_2022}), any $n$-dimensional differential system of the form: 
$$\dot x=f(t,x)$$
with Caratheodory solutions can be represented in a harmonic form as follows:
$$\dot X=\mathcal{F}(f(t,x))-\mathcal{N}X$$
where \begin{equation}\mathcal{N}:=\mathrm{Id}_n\otimes \mathrm{diag}( \textsf{j}\omega k,\ k\in \mathbb{Z})\label{q}\end{equation}
For polynomial systems, $\mathcal{F}(f(t,x))$ can be easily computed using the arithmetic rules \eqref{matvectproduit} and  \eqref{produit}. Furthermore, the recovery of $x$ from $X$ can be achieved through the exact formula~\cite{blin_necessary_2022}:
\begin{equation}\label{recos} 
x(t)=\mathcal{F}^{-1}(X)(t):=\sum_{k=-\infty}^{+\infty} X_k(t)e^{ \textsf{j}\omega k t}+\frac{T}{2}\dot X_0(t)
\end{equation}
In the context of Linear Time Periodic (LTP) systems characterized by $T$-periodic matrix functions, denoted as $A(\cdot)$ and $B(\cdot)$, belonging to the respective classes of $L^2([0, T],\mathbb{C}^{n\times n})$ and $L^{\infty}([0, T],\mathbb{C}^{n\times m})$, we have the following equivalence: Let $x$ is be a solution to 
\begin{align}\dot x(t)=A(t)x(t)+B(t)u(t)\quad x(0):=x_0\label{ltp}\end{align}
associated to the control $u\in L_{loc}^2(\mathbb{R},{\mathbb{C}^m)}$ then, $X=\mathcal{F}(x)$ is a solution to the linear time-invariant (LTI) system 
\begin{align}
	\dot X(t)=(\mathcal{A}-\mathcal{N})X(t)+\mathcal{B}U(t), \quad X(0):=\mathcal{F}(x)(0) \label{ltih}
\end{align}
associated to $U=\mathcal{F}(u)$ with $\mathcal{A}:=\mathcal{T}(A)$, $\mathcal{B}:=\mathcal{T}(B)$. 
Reciprocally, if $X\in H$ is a solution to \eqref{ltih} associated to $U\in H$, then $x$ is a solution to~\eqref{ltp} associated to $u$ (i.e. $X:=\mathcal{F}(x)$ and $U:=\mathcal{F}(u)$). 

\section{Framework for harmonic robust control}
\subsection{Harmonic transfer and transition functions}
Consider a signal $X\in L^2(\mathbb{R}^+,\ell^2(\mathbb{C}^n))$ with its induced norm
$$\|X\|^2_{L^2}:=\int_0^{+\infty}\|X(t)\|^2_{\ell^2}dt$$
and its Laplace transform:
$$\hat X(s):=\int_0^{+\infty}X(t)e^{-st}dt$$ 
The associated $H_2$ induced norm is given by
$$\|\hat X\|_{H_2}:=\left(\frac{1}{2\pi}\int_{-\infty}^{+\infty}\|\hat X(\textsf{j} \omega)\|^2_{\ell^2}d\omega\right)^{\frac{1}{2}}.$$
By Paley Wiener theorem, we have : $\|\hat X\|_{H_2}= \|X\|_{L^2}$.
\begin{definition}For harmonic system \eqref{ltih} with output  $Y:=\mathcal{C}X$, the harmonic transfer function between input $\hat U(s)$ and output $\hat Y(s)$ is given by:
\begin{equation}
\hat{\mathcal{G}}(s):=\mathcal{C}(s\mathcal{I}-(\mathcal{A}-\mathcal{N}))^{-1}\mathcal{B}\label{Laplace_G}
\end{equation} and is the Laplace transform of 
\begin{equation}\mathcal{G}(t):=\mathcal{C}e^{(\mathcal{A}-\mathcal{N})t}\mathcal{B}.\label{impulse_h}\end{equation}
\end{definition}

As shown in \cite{beam1993asymptotic, zhou2004spectral}, the spectrum denoted as $\sigma$ of the operator $(\mathcal{A}-\mathcal{N})$ in \eqref{ltih} exclusively comprises eigenvalues forming an unbounded, discrete set that relies solely on a finite number of complex values $\lambda_i$, $i=1,\cdots,n$: $$\sigma:=\{\lambda_i+\textsf{j}\omega k,k\in\mathbb{Z}, i=1,\cdots,n\}.$$
For the sake of clarity and simplicity in our exposition, we make the assumption that $(\mathcal{A}-\mathcal{N})$ is non-defective. Under this assumption, the following eigenvalue decomposition, as described in \cite{riedinger2022solving}, takes place:
\begin{equation}(\mathcal{A}-\mathcal{N})\mathcal{V}=\mathcal{V}(\Lambda\otimes\mathcal{I}-\mathcal{N})\label{decomp}\end{equation}
with $\Lambda=diag(\lambda_i, i=1, \cdots,n)$ and where $\mathcal{V}$ is a constant, invertible TB and bounded operator on $\ell^2$.
Notice that in case where $(\mathcal{A}-\mathcal{N})$ is defective,
the diagonal matrix $\Lambda$ can be substituted with a Jordan canonical form denoted as $J$ (as outlined in \cite{riedinger2022solving}). 

The following property will prove to be valuable in the sequel.
 \begin{property}For any real number $\delta$, if $z(t):=x(t-\delta)\ a.e$  where $x$ is a vector function ($\in L_{loc}^2$) then $$Z(t):=e^{-\mathcal{N}\delta}X(t-\delta),$$ where $Z:=\mathcal{F}(z)$ and $X:=\mathcal{F}(x)$ and where $\mathcal{N}$ is given by Eq. \eqref{q}.
Moreover, if $M(t):=G(t-\delta) \ a.e.$ where $G$ is a matrix function ($\in L_{loc}^\infty$), then  $$\mathcal{M}(t):=e^{-\mathcal{N}\delta}\mathcal{G}(t-\delta)e^{\mathcal{N}\delta}$$
where $\mathcal{M}:=\mathcal{T}(M)$ and $\mathcal{G}:=\mathcal{T}(G)$.
 \end{property}
 \begin{proof}
It is straightforward to show that $Z_k(t)=e^{-\textsf{j}\omega k \delta}X_k(t-\delta)$ from which we can deduce the results of this property.
 \end{proof}

\begin{theorem} For any time instant $t$, the exponential of the harmonic state operator in \eqref{ltih}, denoted as 
$e^{(\mathcal{A}-\mathcal{N})t}$, is a bounded operator on $\ell^2$ and it is not TB. Furthermore, the state transition function associated with \eqref{ltp} for the time interval between $t_0$ and $t$ is expressed as follows:
 \begin{equation}\Phi(t,t_0):=V(t)e^{\Lambda (t-t_0)}V^{-1}(t_0)\label{tf_ltp}\end{equation}
where the matrix function $V$ is a $T$-periodic, invertible, and absolutely continuous function, defined as $V := \mathcal{T}^{-1}(\mathcal{V})$, where $\mathcal{V}$ is given by \eqref{decomp}.
\end{theorem} 
\begin{proof} For any $t$, we have:
\begin{align*}
\|e^{(\Lambda\otimes\mathcal{I}-\mathcal{N})t}\|_{\ell^2}&:=\sup_{\|X\|_{\ell^2}=1}\|e^{(\Lambda\otimes\mathcal{I}-\mathcal{N})t}X\|_{\ell^2}=\bar \sigma(t)<+\infty
\end{align*} where $\bar \sigma(t):=\max_{i=1,\cdots,n}e^{2Re(\lambda_i)t}$. Here, we use the fact that $e^{\mathcal{-N}^*t}e^{\mathcal{-N}t}=\mathcal{I}$  which implies $e^{(\Lambda\otimes\mathcal{I}-\mathcal{N})^*t}e^{(\Lambda\otimes\mathcal{I}-\mathcal{N})t}=e^{((\Lambda^*+\Lambda)\otimes\mathcal{I})t}$. 
Using \eqref{decomp}, we have that $e^{(\mathcal{A}-\mathcal{N})t} $ is bounded on $\ell^2$ for every $t$.   Moreover, $e^{(\mathcal{A}-\mathcal{N})t}$ is obviously not TB due to the non Toeplitz diagonal term $\mathcal{N}$. In addition, the free response of \eqref{ltp} between time $t_0$ and $t$ from an initial condition $x(t_0)$ is given by $x(t)=\Phi(t,t_0)x(t_0)$ and its associated harmonic free response is given by:
$$X(t)=e^{(\mathcal{A}-\mathcal{N})(t-t_0)}X(t_0)=\mathcal{V}e^{(\Lambda\otimes\mathcal{I})(t-t_0)}e^{-\mathcal{N}(t-t_0)}\mathcal{V}^{-1}X(t_0)$$
As $ \mathcal{V}^{-1}X(t_0)=\mathcal{V}^{-1}e^{\mathcal{N}(t-t_0)}e^{-\mathcal{N}(t-t_0)}X(t_0)$, using Property 1 and rules \eqref{matvectproduit} and \eqref{produit} leads to: \begin{align*}
    \Phi(t,t_0)x(t_0)&=\mathcal{F}^{-1}(e^{(\mathcal{A}-\mathcal{N})(t-t_0)}X(t_0))\\&=V(t)e^{\Lambda(t-t_0)}V^{-1}(t_0)x(t_0)
\end{align*}
for any $x(t_0)$. The last result is established noticing that the absolute continuity of $V$ is proved in \cite{riedinger2022solving}.
\end{proof}
\begin{remark} It is essential to emphasize that $\mathcal{T}(\Phi(t,t_0)) = e^{\mathcal{A}(t-t_0)}$, and this should not be confused with the harmonic transition function $e^{(\mathcal{A}-\mathcal{N})(t-t_0)}$.
\end{remark}
%The response of \eqref{ltih} from $X_0 \in\ell^2$ and for input $U\in L^2(\mathbb{R}^+,\ell^2(\mathbb{C}^n))$, given by \begin{equation}X(t)=e^{(\mathcal{A}-\mathcal{N})t}X_0+\int_0^te^{(\mathcal{A}-\mathcal{N})(t-\tau)}\mathcal{B}U(\tau)d\tau\label{trans}\end{equation} belongs to $\ell^2$ for any $t$. Moreover if $X$ and $U$ belongs to $H$, then there exist $u:=\mathcal{F}^{-1}(U)$ and $x:=\mathcal{F}^{-1}(X)$ such that
%\begin{equation}
%x(t)=\Phi(t,0)x_0+\int_0^t\Phi(t,s)B(s)u(s)ds\label{trans2}\end{equation} where $\Phi(t,s)$ refers to the transition matrix function of system~\eqref{ltp}. 

\subsection{Average Trace operator}
TB operators are not a trace-class. Indeed, for a $n\times n$ TB bounded operator on $\ell^2$, we have
\begin{align}tr(|\mathcal{M}|)&:=\sum_{k\in\mathbb{Z} }<e_{k},(\mathcal{M}^*\mathcal{M})^\frac{1}{2}e_{k}>=\sum_{k\in\mathbb{Z}}\sum_{i=1}^nm_i=+\infty\end{align}
 where $m_i$ denote the $n$ diagonal values of $(\mathcal{M}^*\mathcal{M})^\frac{1}{2}$ and $(e_{k})$ is an orthonormal basis for $(\ell^2)^n$.  Nevertheless, a solution to circumvent this challenge consists in considering the vector space of matrix functions in $L^\infty([0, T])$ endowed with the conventional scalar product: 
\begin{align}
	<M,N>&:=\frac{1}{T}\int_0^Ttr(M'(\tau)N(\tau))d\tau\label{sca}\\
&=\frac{1}{T}\sum_{i,j=1}^n\int_0^TM_{ij}(\tau)N_{ij}(\tau) d\tau,\label{sca2}
\end{align}
for which it is straightforward to show that the induced norm\footnote{Frobenius norm: $\|M\|_{F}:=<M,M>^\frac{1}{2}$ } satisfies:\begin{equation}\|M\|_{L^2([0 \ T])}\leq\|M\|_F\leq \sqrt{n}\|M\|_{L^\infty([0 \ T])}.\label{nl1}\end{equation}
%We recall (see p.p. 562-574 of \cite{gohberg1993classes} for a detailed proof): 
%\begin{theorem}\label{Linfell2}
%$M\in L^\infty([0\ T])$ if and only if $\mathcal{M=}\mathcal{T}(M)$ is bounded on $\ell^2$  i.e. there exists $\kappa>0$, $$\|\mathcal{M}\|_{\ell^2}=\sup_{\|x\|_{\ell^2}=1}\|\mathcal{M}x\|_{\ell^2}<\kappa,$$  
%and that the following equality occurs:
%\begin{equation}\|\mathcal{M}\|_{\ell^2}=\|M\|_{L^\infty}.\label{nl2}\end{equation}
%\end{theorem}

As any component of $M$ can be rewritten using its Fourier series (since $L^\infty([0\ T])\subset L^2([0\ T])$:
$$M_{ij}(t)=\sum_{k\in \mathbb{Z}} m_{ij,k} e^{ \textsf{j}\omega kt}\ a.e.,$$
 it follows, using \eqref{sca2}, that: $$<Id,M>=\sum_{i=1}^n m_{ii,0}.$$
This allows to define the average-trace operator as follows.
\begin{definition}\label{trace}The average trace operator for TB bounded operators on $\ell^2$ is defined by
	\begin{equation}
		tr_0(\mathcal{M}):=\sum_{i=1}^n m_{ii,0}.\label{tr}\end{equation}
		where $m_{ii,0}$ refers to the diagonal value located on the $i-th$ diagonal block of $\mathcal{M}$.
\end{definition}
\begin{theorem}\label{normtr}For any constant, $n\times n$ TB and bounded operator on $\ell^2$, denoted by $\mathcal{M}$, $tr_0(\mathcal{M}^*\mathcal{M})^\frac{1}{2}$ defines an operator-norm that satisfies: \begin{equation}
\|\mathcal{M}\|_{\ell^2}\leq tr_0(\mathcal{M}^*\mathcal{M})^\frac{1}{2}=\|M\|_F\leq \sqrt{n} \|\mathcal{M}\|_{\ell^2}\label{equiv}\end{equation}
where $M:=\mathcal{T}^{-1}(\mathcal{M})$.
\end{theorem}
\begin{proof}
The result is deduced from \eqref{nl1},  Theorem \ref{borne} and recalling that Riesz-Fischer Theorem implies:
	\begin{align*}
			\|M\|_{L^2([0 \ T]}&:=\sup_{\|x\|_{L^2}=1}\|Mx\|_{L^2(0 \ T])}\\
   &=\sup_{\|X\|_{\ell^2}=1}\|\mathcal{M}X\|_{\ell^2}:=\|\mathcal{M}\|_{\ell^2}
	\end{align*} where $X:=\mathcal{F}(x)$ and $\mathcal{M}:=\mathcal{T}(M)$.
\end{proof}
{\remark In Definition~\ref{trace}, the average-trace of a TB operator can be interpreted as a mean value of the trace operator, that is  
\begin{align}
tr_0(\mathcal{M})&:=\lim_{N\rightarrow+\infty}\frac{1}{2N+1}\sum_{|k|\leq N}\sum_{i=1}^n<e_{ki},\mathcal{M}_{ii} e_{ki}>\\
&=\lim_{N\rightarrow+\infty}\frac{1}{2N+1}\sum_{|k|\leq N}\sum_{i=1}^nm_{ii,0}=\sum_{i=1}^nm_{ii,0}<+\infty\nonumber\end{align} where $\mathcal{M}_{ii}$ refers to the $ith$ Toeplitz diagonal block of $\mathcal{M}$ and $(e_{ki})$ is an orthonormal  basis of $\ell^2$ associated to the  $ith$ diagonal block.}

%\begin{remark} A situation where 
%Definition~\ref{trace} is useful is the following. Suppose we have to determine a Hermitian, positive definite, TB and bounded operator on $\ell^2$ denoted as $\mathcal{M}$, while adhering to specific constraints aimed at minimizing $X^*\mathcal{M}X$ for any $X\in\ell^2$.
%In light of the existence of such an operator, it implies that for any basis vector $e_k$ within $\ell^2$, the term $e_k^*\mathcal{M}e_k$ attains its minimum, and consequently, it follows that $tr_0(\mathcal{M})$ is minimized. In the sequel, we will employ this average trace operator as an alternative approach to circumvent the issue of an infinite trace value in $tr(\mathcal{M})$.
%\end{remark}

%Let $S^n_+=\{M\in S^n: M\geq 0\ a.e.\}$ and $S^n_{++}=\{M\in S^n: M> 0\ a.e.\}$. 
%It follows that: $M\in S^n_{++}$ if and only if $\mathcal{M}=\mathcal{T}(M)$ is Hermitian, positive definite, TB and bounded on $\ell^2$ {\cite{blin_necessary_2022,zhou_stability_2001}}. 

\subsection{Harmonic $H_2$ and $H_\infty$ norms}
In the harmonic framework, defining the harmonic $H_2$ norm is not a straightforward task, primarily because the Frobenius norm of  $\hat{\mathcal{G}}(\textsf{j}\omega)$, given by \eqref{Laplace_G}, $$\|\hat{\mathcal{G}}(\textsf{j}\omega)\|^2_F:=tr(\hat{\mathcal{G}}(\textsf{j}\omega)^*\hat{\mathcal{G}}(\textsf{j}\omega))$$ is not trace-class. To address this challenge, let us first introduce the following result.
\begin{theorem}\label{toeplitzM}For any time instant $t$,
\begin{equation}
\mathcal{M}(t):=\mathcal{G}(t)^*\mathcal{G}(t)\end{equation}
 where $\mathcal{G}(t):= \mathcal{C}e^{(\mathcal{A}-\mathcal{N})t}\mathcal{B}$,
 is a hermitian and non-negative, TB and bounded operator on $\ell^2$.
 Moreover, $M:=\mathcal{T}^{-1}(\mathcal{M})$ exists and is given at time $\tau$ by:
\begin{equation}
M_t(\tau):=G(\tau,t)'G(\tau,t)\label{impulse_t}\end{equation}
where $G(\tau,t):=C(\tau+t)\Phi(\tau+t,\tau)B(\tau)$.
 \end{theorem}
 \begin{proof} Hermitian and non-negativity properties are readily evident for $\mathcal{M}(t)$. 
Using the decomposition \eqref{decomp}, we have:
\begin{align}
e^{(\mathcal{A}-\mathcal{N})^*t}\mathcal{C}^*\mathcal{C}e^{(\mathcal{A}-\mathcal{N})t}=\mathcal{V}^{-1*}e^{(\Lambda\otimes\mathcal{I})^*t}e^{-\mathcal{N}^*t}\mathcal{H}e^{-\mathcal{N}t}e^{(\Lambda\otimes\mathcal{I})t}\mathcal{V}^{-1}\label{par1}
\end{align} where $\mathcal{H}=\mathcal{V}^*\mathcal{C}^*\mathcal{C}\mathcal{V}$.
Thus, as $e^{-\mathcal{N}^*t}\mathcal{H}e^{-\mathcal{N}t}=e^{\mathcal{N}t}\mathcal{H}e^{-\mathcal{N}t}$, Property 1 implies that $e^{-\mathcal{N}^*t}\mathcal{H}e^{-\mathcal{N}t}$ is obviously a TB operator since it is equal to $\mathcal{T}(L)$ where \begin{equation}L(\tau):=H(\tau+t)\label{Ltau}\end{equation} for almost all $\tau$ with $H:=\mathcal{T}^{-1}(\mathcal{H})$.
%The TB property can be deduced thanks to Property 1 since cal{H}$,
 %$$\mathcal{Z}:=e^{\mathcal{N}^*t}\mathcal{H} e^{ \mathcal{N}t}$$ is a TB operator whose components satisfy: $Z_{ij,k}=e^{-\textsf{j}\omega kt}H_{ij,k}$ for $i,j=1,\cdots,n$.
It follows that $\mathcal{M}(t)$ is a product of TB operators and thus a TB operator (see \eqref{produit}).

Finally as $e^{(\mathcal{A}-\mathcal{N})t}$ is also bounded on $\ell^2$ for every $t$ (its maximal singular value is $e^{\sigma t}$ where $\sigma :=\max_i 2Re(\lambda_i)$), $\mathcal{M}(t)$ is also a bounded operator on $\ell^2$.
For the last assertion,  using Property 1,  we see that \eqref{par1} is the Toeplitz transformation of $Z(\tau)$ where $Z(\tau):= V^{'-1}(\tau)e^{\Lambda^* t}L(\tau)e^{\Lambda t}V^{-1}(\tau)$ with $L$ defined by \eqref{Ltau}.
As $H(\tau):=V'(\tau)C'(\tau)C(\tau)V(\tau)$, it follows, using \eqref{tf_ltp}, that $Z(\tau)$ can be rewritten as:
$$Z(\tau):= \Phi'(\tau+t,\tau)C'(\tau+t)C(\tau+t)\Phi(\tau+t,\tau)$$ 
Consequently, multiplying $Z(\tau)$ on the left by $B'(\tau)$ and on the right by $B(\tau)$, we get the final result (thanks to \eqref{produit}). 
\end{proof}
Given the assumption that $(\mathcal{A}-\mathcal{N})$ is Hurwitz, we define the following induced operator-norm for: $\mathcal{G}(t):=\mathcal{C}e^{(\mathcal{A}-\mathcal{N})t}\mathcal{B}$
$$\|\mathcal{G}\|_{L^2}:=\left(\int_{0}^{+\infty} tr_0(\mathcal{G}^*(t)\mathcal{G}(t)) dt\right)^\frac{1}{2}$$
where $tr_0$ is given in Definition \ref{trace} for TB and bounded operators on $\ell^2$.
Consequently, as $tr_0(\mathcal{G}^*(t)\mathcal{G}(t))^\frac{1}{2}$ is finite for every $t$ and is clearly in $L^2$, we are ready to define the harmonic $H_2$ norm.
%\color{red}From Definition \ref{proj} of  $\Pi_m$, this can be explicitly rewritten:
%$$\|\mathcal{G}\|_{L^2}=\left(\int_{0}^{+\infty} tr(\Pi_0(\mathcal{G}(t)\mathcal{G}^*(t)))  dt\right)^\frac{1}{2}$$
%\color{black}

 \begin{definition}The $H_2$ operator norm of the harmonic transfer function $\hat{\mathcal{G}}(s)$ is defined by
 \begin{align}\|\hat{\mathcal{G}}(s)\|_{H_2}&:=\left(\frac{1}{2\pi}\int_{-\infty}^{+\infty} tr_0(\hat{\mathcal{G}}(\textsf{j}\omega)^*\hat{\mathcal{G}}(\textsf{j}\omega))d\omega
 \right)^\frac{1}{2}\end{align}
 \end{definition}
 The next theorem states that the $L^2$ norm of the sliding average Frobenius norm of $G(\tau,t)$ over a window of length $T$ is equal to the $H_2$ operator norm of its associated harmonic transfer function.

 \begin{theorem}The following equalities hold true:
 $$\|G\|_{L^2}=\|\mathcal{G}\|_{L^2}=\|\hat{\mathcal{G}}(s)\|_{H_2}$$
where $\hat{\mathcal{G}}$, $\mathcal{G}$ and $G$ are defined respectively by \eqref{Laplace_G}, \eqref{impulse_h} and \eqref{impulse_t} and where \begin{align}
\|G (t)\|_{L^2}&:=\left(\int_{0}^{+\infty} \|G(\tau,t)\|_F^2 dt\right)^\frac{1}{2}\\
&=\left(\int_0^{+\infty}\frac{1}{T}\int^t_{t-T}tr(G(\tau,t)'G(\tau,t))d\tau dt\right)^\frac{1}{2}\\
\|\mathcal{G}\|_{L^2}&:=\left(\int_{0}^{+\infty} tr_0(\mathcal{G}^*(t)\mathcal{G}(t)) dt\right)^\frac{1}{2}\\
\|\hat{\mathcal{G}}(s)\|_{H_2}&:=\left(\frac{1}{2\pi}\int_{-\infty}^{+\infty} tr_0(\hat{\mathcal{G}}(\textsf{j}\omega)^*\hat{\mathcal{G}}(\textsf{j}\omega))d\omega \right)^\frac{1}{2}
 \end{align}
 \end{theorem}
 \begin{proof} The first equality follows by Theorems~\ref{normtr} and \ref{toeplitzM} while the second one is due to Paley-Wiener Theorem.
 \end{proof}

%Consider signal $X\in L^2(\mathbb{R}^+,\ell^2(\mathbb{C}^n))$ with induced norm
%$$\|X\|^2_{L^2}=\int_0^{+\infty}\|X(t)\|^2_{\ell^2}dt$$
%and let us define its Laplace transform as follows:
%$$\hat X(s)=\int_0^{+\infty}X(t)e^{-st}dt$$ 
%with a region of convergence given by $Re(s) > a$ for some $a$
%within the associated $H_2$ induced norm is given by
%$$\|\hat X\|_{H_2}=\left(\frac{1}{2\pi}\int_{-\infty}^{+\infty}\|\hat X(j \omega)\|^2_{\ell^2}d\omega\right)^{\frac{1}{2}}.$$
%
%Then, by Paley Wiener theorem, we have : $\|\hat X\|_{H_2}= \|X\|_{L^2}$.
%
Now, let us define the $H_\infty$ norm in the harmonic context. 
\begin{definition}
The $H_\infty$ operator norm  of the harmonic transfert function $\hat{\mathcal{G}}(s)$ is defined by:
\begin{align}\|\hat{\mathcal{G}}(s)\|_{H_\infty}&:=ess \sup \bar\sigma(\hat{\mathcal{G}}(\textsf{j}\omega))
\end{align}
where $\bar\sigma$ refers to the maximal singular value.
\end{definition}
The next theorem states that the $H_\infty$ operator norm of the harmonic transfer function defines a bound for the $L^2$ norm of the sliding quadratic average of the output $y$ over a window of length $T$. 
\begin{theorem}The following relations hold true:
\begin{align}\|\hat{\mathcal{G}}(s)\|_{H_\infty}&:=\sup_ {\|\hat U\|_{L^2}=1}\|\hat{\mathcal{G}}\hat U\|_{L^2}\label{eee2}\\
&=\sup_ {\| U\|_{L^2}=1}\|\mathcal{G}*U\|_{L^2}\label{eee3}\\
&\geq \sup_ {\| U\|_{L^2\cap H}=1}\|\mathcal{G}*U\|_{L^2}\label{eee4}\\
&=\sup_ {\|u\|_{L^2}=1}(\int_0^\infty\|y(t)\|^2_{L^2([t-T\ t])}dt)^\frac{1}{2}\label{eee5}
\end{align} where $*$ denotes the convolution product and $y(t):=C(t)\int_0^t\Phi(t,s)B(s)u(s)ds$. \end{theorem}
\begin{proof}Relations \eqref{eee2} and \eqref{eee3} are obvious. To prove \eqref{eee5}, 
if $U\in {L^2}\cap H$ then $u:=\mathcal{F}^{-1}(U)$ exists and it follows that $y:=\mathcal{F}^{-1}(Y)$ with $Y:=\mathcal{G}*U$ satisfies  $$y(t):=\int_0^tC(t)\Phi(t,s)B(s)u(s)ds$$ where $\Phi(t,s)$ denotes the transition matrix associated to \eqref{ltp}.

Invoking Riesz-Fischer Theorem, if $W:=\mathcal{F}(w)$ then $$\|w\|_{L^2([t-T\ t])}(t):=(\frac{1}{T}\int^t_{t-T}|w(\tau)|^2d\tau)^\frac{1}{2}=\|W(t)\|_{\ell^2}.$$ It follows that the $L^2$ norm of the sliding quadratic average of $w$ over a window of length $T$ defined as
$$\|w\|_{L^2}:=(\int_0^\infty\ \|w\|^2_{L^2([t-T\ t])}dt)^\frac{1}{2}$$
fulfils the relation:
$$\|w\|_{L^2}=(\int_0^\infty\ \|W(t)\|_{\ell^2}^2dt)^\frac{1}{2}=\| W\|_{L^2}.$$
Therefore, we have: 
\begin{align}
\|\hat{\mathcal{G}}(s)\|_{H_\infty}&\geq \sup_ {\| U\|_{L^2\cap H}=1}\|\mathcal{G}*U\|_{L^2}\\
&=\sup_ {\| U\|_{L^2\cap H}=1}(\int_0^\infty \|Y(t) \|_{\ell^2}^2dt)^\frac{1}{2}\\
&= \sup_ {\|u\|_{L^2}=1}(\int_0^\infty\|y(t)\|^2_{L^2([t-T\ t])}dt)^\frac{1}{2}\end{align}
\end{proof}

Having introduced the harmonic $H_\infty$ and $H_2$ norms, our aim in the sequel is to demonstrate how the established semi-definite convex optimization techniques for tackling robust control problems in finite-dimensional LTI systems can be extended into the harmonic framework. 

\subsection{Harmonic $H_2$ and $H_\infty$ optimal control}
Consider a LTP system with control input $u$, exogenous input $w$, controlled output $z$ and measured output $y$ given by 
\begin{align}
	\dot x(t)&=A(t)x(t)+B(t)u(t)+B_w(t)w(t),\label{ltpr}\\
z(t)&=C_z(t)x(t)+D_{zw}(t)w(t)+D_{zu}(t)u(t)\\
	y(t)&=C_y(t)x(t)+D_{yw}(t)w(t)+D_{yu}(t)u(t)
\end{align}
where all inputs belong to $L^2(\mathbb{R}^+)$ and all matrices are assumed to belong to $L^\infty([0 \ T])$.
For each of these matrices, the associated infinite dimensional TB matrix function is obtained using (\ref{ToepFormMat}). 
This allows to write the harmonic LTI representation as follows:
\begin{align}
	\dot X(t)&=(\mathcal{A}-\mathcal{N})X(t)+\mathcal{B}U(t)+\mathcal{B}_wW(t),\label{ltir}\\
	Z(t)&=\mathcal{C}_zX(t)+\mathcal{D}_{zw}W(t)+\mathcal{D}_{zu}U(t)\\
	Y(t)&=\mathcal{C}_yX(t)+\mathcal{D}_{yw}W(t)+\mathcal{D}_{yu}U(t)
\end{align}
From this perspective, it becomes evident that any robust optimal control problem formulated for finite-dimensional LTI systems can be naturally extended to its infinite-dimensional harmonic counterpart. Here, we focus on the full-state feedback $H_2$ and $H_\infty$ problems for which $\mathcal{C}_y=\mathcal{I}, \mathcal{D}_{zw}=\mathcal{D}_{yw}=\mathcal{D}_{yu}=0$. For instance, consider the state feedback harmonic $H_2$ optimal control problem, which consists in determining  $U=\mathcal{K}X$, where $\mathcal{K}$ is a bounded operator on $\ell^2$, to minimize the $H_2$ norm of the harmonic transfer function between $W$ and $Z$. The optimal solution is given by $\mathcal{K}=\mathcal{S}\mathcal{P}^{-1}$ where $\mathcal{P},$ $\mathcal{S}$ and $\mathcal{Z}$ are TB and bounded operators on $\ell^2$ that solve the following convex optimization problem:  
\begin{align}
&\min_{\mathcal{P},\mathcal{S},\mathcal{Z}} ~~tr_0(\mathcal{Z}) \label{H222}\\
&\text{ subject to:}\nonumber\\
& \mathcal{P}^*=\mathcal{P}>0 \nonumber\\
&\left[\begin{array}{cc}\mathcal{Z} & \mathcal{C}_z\mathcal{P}+\mathcal{D}_{zu}\mathcal{S} \\ \mathcal{P}\mathcal{C}_z^*+\mathcal{S}^*\mathcal{D}_{zu}^* & \mathcal{P}\end{array}\right]\geq 0\label{PB_H2}\\
&\left[\begin{array}{c}(\mathcal{A}-\mathcal{N})\mathcal{P}+\mathcal{P}(\mathcal{A}-\mathcal{N})^*+\mathcal{BS}+ \mathcal{S}^*\mathcal{B}^*+ \mathcal{B}_w\mathcal{B}_w^*\end{array}\right]\leq 0\nonumber
\end{align}

In a similar vein, the solution to the state feedback harmonic $H_\infty$ optimal control problem takes the form of $U=\mathcal{K}X$, where $\mathcal{K}:=\mathcal{S}\mathcal{P}^{-1}$ and $\mathcal{P},$ $\mathcal{S}$ are TB and bounded operators on $\ell^2$ that solve the following convex optimization problem:  
\begin{align}
&\min_{\mathcal{P},\mathcal{S},\gamma} \gamma \text{ subject to:}\nonumber\\
& \mathcal{P}^*=\mathcal{P}>0\label{Hinf}\\
&\left[\begin{array}{ccc}(\mathcal{A}-\mathcal{N})\mathcal{P}+\mathcal{P}(\mathcal{A}-\mathcal{N})^*+\mathcal{BS}+ \mathcal{S}^*\mathcal{B}^* & \star & \star\\
\mathcal{B}_w^* & -\gamma\mathcal{I}  & \star\\
\mathcal{C}_{z}\mathcal{P}+\mathcal{D}_{zu}\mathcal{S}& \mathcal{D}_{zw} &-\gamma \mathcal{I}\end{array}\right] \leq 0 \nonumber\end{align}

%In practice, to compute $\mathcal{K}$, assuming that $\mathcal{P}$ is known and thus $P$, $\mathcal{P}^{-1}$ can be computed from $P^{-1}$ since $\mathcal{P}^{-1}=\mathcal{T}(P^{-1})$.\\

The primary challenge in the preceding convex optimization problems lies in their infinite-dimensional nature. In practical terms, achieving a solution that facilitates the construction of a control gain necessitates truncation, but it must be carried out consistently. By "consistency", we mean that the approximate solution at a given truncation order should converge to the solution of the original problem as the order increases. Without this consistency, there can be no assurance of optimality or stability with the obtained approximate solution.

In the upcoming sections, our objective is to introduce a precise definition of a consistent truncation scheme. This scheme should enable us to solve the optimal robust control problem with an arbitrarily small margin of error. To achieve this, we will delve into the concept of TBLMIs and shed light on their relationship with differential LMIs in the time domain.

\color{black}

\section{TBLMIs vs PDLMIs}
%\begin{definition}
A TBLMI is an infinite dimensional LMI having the form:
\begin{equation}\mathcal{M}(x)=\mathcal{M}_0 +\sum_{i=1}^{+\infty}x_i\mathcal{M}_i >0,\label{tblmi}
\end{equation}
where $x_i\in \mathbb{R}$, $i=1,2,\cdots $ are the unkown variables and where the Hermitian matrices $\mathcal{M}_i$ are $n\times n$ infinite-dimensional TB  operators.  
A sequence $x=(x_1,x_2,\cdots)$ is a solution of the TBLMI (\ref{tblmi}) if $\mathcal{M}(x)$ is a positive definite and bounded operator on $\ell^2$ i.e. for any $y\in \ell^2$, $y^*\mathcal{M}(x)y>0$ and there exists $C>0$ s.t. $\|\mathcal{M}(x)\|_{\ell^2}\leq C$.
%\end{definition}
As it is the case for LMIs, we often encounter problems in which the variables are matrices, e.g., the harmonic Lyapunov inequality that allows to analyse stability of the harmonic model (\ref{ltih})
\begin{equation}(\mathcal{A}-\mathcal{N})^*\mathcal{P}+\mathcal{P}(\mathcal{A}-\mathcal{N})<0\label{lyap}
\end{equation}
where $\mathcal{A}$ and $\mathcal{N}$ are given and $\mathcal{P}$ is the variable (bounded on $\ell^2$ with $\mathcal{P}=\mathcal{P}^*$).
Obviously, \eqref{lyap} can be rewritten in the form (\ref{tblmi}) by considering a basis $\mathcal{V}_i,$ $i=0,\cdots, +\infty$ for Hermitian and 
 $n\times n$ TB matrices as follows: \begin{equation}\mathcal{M}_i=(\mathcal{A}-\mathcal{N})^*\mathcal{V}_i+\mathcal{V}_i(\mathcal{A}-\mathcal{N}) \text{ and }\mathcal{P}=\sum_{i=0}^{+\infty}x_i\mathcal{V}_i.\label{Fi}
\end{equation}
For the sake of simplification, consider the case where $n=1$ and let $\mathcal{I}_k$ be the $k$-shifted identity matrix with $k \in {\mathbb N}$. Then, a basis denoted as ${\mathcal{V}_i}$ can be obtained from the following set:
 $$\{(\mathcal{I}_{-k}+\mathcal{I}_k),\ {\rm{j}}(-\mathcal{I}_{-k}+\mathcal{I}_k): k \in {\mathbb N}\}.$$
 Additionally, it's worth noting that for any index $i$, $\mathcal{M}_i$ is a TB operator as $\mathcal{N}^*\mathcal{V}_i+\mathcal{V}_i\mathcal{N}$ is a TB operator and the product of TB operators results in a TB operator. 
 
 In the sequel, we consider the following definition:
\begin{definition}
 Let $\mathbb{S}$ be a finite set of subscripts. A TBLMI has an equivalent formulation given by \eqref{tblmi} and is defined by:
\begin{align}
\mathcal{L}(\mathcal{P};\mathcal{A}_s,s\in \mathbb{S})< 0\label{LMgTB}\end{align}
where the entries $\mathcal{A}_s,s\in \mathbb{S}$, are given and $\mathcal{P}$ refers to the TB unknown operator.  
\end{definition} 

The inequality \eqref{lyap} is a TBLMI and corresponds to $\mathcal{L}(\mathcal{P};\mathcal{A})<0$. 
It is important to emphasize that the role of $\mathcal{N}$ in this inequality is not immediately apparent, but its significance will be fully elucidated in the forthcoming theorem. In the time domain, the counterpart of \eqref{lyap} is the  {\it differential} Lyapunov inequality:\begin{equation}\dot P+A'P+PA<0\label{tlyap}
\end{equation} with $P\in L^\infty([0\ T])$ and  $P=P'>0\ a.e.$ An intriguing question arises: Is there a connection between TBLMIs and  PDLMIs? To address this inquiry, we define precisely the notion of PDLMI in the time domain and then establish its equivalence with a TBLMI in the harmonic domain.

\begin{definition}
Let $\mathbb{S}$ be a finite set of subscripts and consider $T-$periodic matrix functions $A_s,$ $s\in \mathbb{S}$, all belonging to $L^\infty$. A PDLMI is defined by:
\begin{equation}L(\dot P(t),P(t);A_s(t),s\in \mathbb{S})< 0\ a.e. \label{dlmi}\end{equation} 
where  the function $L$ is linear with respect to the unknowns $P$ and $\dot P$ and is also a polynomial matrix-valued function with respect to its defining entries $A_s,s\in \mathbb{S}$. We assume that the mapping from $t$ to $L(\dot P(t),P(t);A_s(t),s\in \mathbb{S})$ belongs to $L^\infty_{loc}$, implying that for any vector function $x \in L^2([t-T, t])$,  the inner product: $$<x,Lx>_{L^2([t-T \ t])}<0$$ 
holds, or equivalently, $L(t)<0\ a.e.$. In simpler terms, this means that $L$ defines a negative definite linear operator on $L^2_{loc}$.
%The negative sign in \eqref{dlmi} means that for every $x\in L^2([0 \ T])$, the inner product satisfies:
%$$<x,L(\dot P,P;A_s,s\in \mathbb{S})x>_{L^2}< 0$$
\end{definition}

A solution, denoted as $P(\cdot): \mathbb{R} \rightarrow \mathbb{R}^{n\times n}$, to \eqref{dlmi}, if it exists, is expected to be an absolutely continuous and symmetric matrix function. The following result establishes the connection between TBLMIs and PDLMIs.

\begin{theorem}\label{dlmitotblmi}
$P(t)$ is a $T-$periodic solution to \eqref{dlmi} if and only if $\mathcal{P}:=\mathcal{T}(P)$ is a constant, hermitian and bounded operator on $\ell^2$ and satisfies the infinite dimensional TBLMI:
\begin{align}
\mathcal{L}(\mathcal{P};\mathcal{A}_s,s\in \mathbb{S})< 0\label{LMg}\end{align} 
with $\mathcal{A}_s:=\mathcal{T}(A_s)$, $s\in \mathbb{S}$, and where 
$$\mathcal{L}(\mathcal{P};\mathcal{A}_s,s\in \mathbb{S}):=L(-\mathcal{N}^*\mathcal{P}-\mathcal{PN},\mathcal{P};\mathcal{A}_s,s\in \mathbb{S}).
$$  
%In addition is $P(t)$ is $T-$periodic then $\dot{\mathcal{P}}=0$ and the above DLMI is just an LMI on $\ell^2$. 
%$\mathcal{L}(\dot {\mathcal{P}},\mathcal{P};\mathcal{A}_i,s\in \mathbb{S})$
\end{theorem}
\begin{proof}If $P$ is absolutely continuous, $\dot P$ exists almost everywhere and is integrable, thus $\mathcal{F}(\dot P)$ can be defined. 
Moreover as $P(t)$ and $L(\dot P(t),P(t);A_s(t),s\in \mathbb{S}) $ are $L^{\infty}_{loc}$, the operators $\mathcal{P}:=\mathcal{T}(P)$ and 
$\mathcal{T}(L(\dot P(t),P(t);A_s(t),s\in \mathbb{S}))$
 are bounded operators on $\ell^2$ (see Theorem~\ref{borne}).
As a LMI is a polynomial matrix valued function with respect to all its entries, the Toeplitz transformation $\mathcal{T}(L)$ of $L$ is obtained by replacing all the terms in the PDLMI by their Toeplitz transformation and $ \dot P$ by $\dot{\mathcal{P}}-\mathcal{N}^*\mathcal{P}-\mathcal{PN}$ (see Proof of Theorem 5 in \cite{blin_necessary_2022} for more detail).
This leads formally to: 
\begin{align}
\mathcal{T}(L(\dot P,&P;A_s,s\in \mathbb{S}))\nonumber\\
&=L(\dot{\mathcal{P}}-\mathcal{N}^*\mathcal{P}-\mathcal{PN},\mathcal{P};\mathcal{A}_s,s\in \mathbb{S})\label{hlmi}
\end{align}
As it is assumed that $P(t)$ is $T-$periodic, it follows that $\mathcal{P}$ is constant. Hence, $\dot {\mathcal{P}}=0$ and \eqref{hlmi} is a constant bounded operator on $\ell^2$.
Moreover, it can be readily demonstrated through the application of the Riesz-Fisher Theorem that for any $x\in L^2([t-T \ t])$:
\begin{align*}
&<x,L(\dot P,P;A_s,s\in \mathbb{S})x>_{L^2([t-T\ t]}(t)\\
&:=\frac{1}{T}\int^t_{t-T}x'L(\dot P,P;A_s,s\in \mathbb{S})x d\tau\\
&=X(t)^*L(\dot{\mathcal{P}}-\mathcal{N}^*\mathcal{P}-\mathcal{PN},\mathcal{P};\mathcal{A}_s,s\in \mathbb{S})X(t)<0
\end{align*} 
where $X:=\mathcal{F}(x)$. 
Finally, replacing $\dot {\mathcal{P}}$ by zero leads to \eqref{LMg}. The converse is obvious since the bounded operators on $\ell^2$, $\mathcal{P}$ and $\mathcal{A}_s,s\in \mathbb{S}$, are constant and thus trivially belong to $H$ (see Theorem~\ref{coincidence}).
\end{proof}
The result of Theorem~\ref{dlmitotblmi} highlights the potential of the harmonic framework to streamline stability analysis and control design synthesis. Indeed, the development of reliable and consistent methods for solving infinite-dimensional TBLMIs presents a promising alternative to the notoriously intricate PDLMIs-based approaches. For example, using Theorem~\ref{dlmitotblmi} and the fact that $tr_0\mathcal{Z}=<Id,Z>$ (see Def.~\ref{trace}),  the time-domain counterpart formulation of the harmonic $H_2$ optimization problem \eqref{H222} is the following differential optimization problem
\begin{align}
&\min_{P, S, Z} <Id,Z> \nonumber\\
& \text{ subject to:}\nonumber\\
&\left[\begin{array}{cc}Z & C_zP+D_{zu}S \\ PC_z'+S'D_{zu}' & P\end{array}\right]\geq 0\\
&\left[\begin{array}{c}\dot P+AP+PA'+BS+ S'B'+ B_wB'_w\end{array}\right]\leq0\nonumber
\end{align}
where $P=P',$ $S$ and $Z$ are $T-$periodic and $L^\infty([0 \ T])$ matrix functions. Given the difficulty associated with solving differential LMIs, as outlined in \cite{colaneri2000continuous}, it becomes especially pertinent to consider the harmonic formulation \eqref{H222}. This becomes all the more relevant as we delve into the problem's infinite-dimensional aspects, which is precisely the focus of the upcoming section.

\section{Solving infinite dimensional harmonic semi definite convex optimization problem}
Considering the efficacy of convex optimization techniques in addressing the majority of optimal robust control problems for linear systems, our primary focus in this section lies in obtaining an approximate solution for the following Convex Optimization Problem ({\bf COP}):
\begin{align*}{\bf COP:} &\min_{\mathcal{P}^*=\mathcal{P}>0} tr_0(\mathcal{P})\text{ subject to:}\\
	&\mathcal{L}(\mathcal{P};{\mathcal{A}_{s}, s\in \mathbb{S}})\leq 0
\end{align*}

We assume that this convex optimization problem is feasible and that the optimal solution is unique, bounded on $\ell^2$ and continuous with respect to the entries $\mathcal{A}_s, s\in\mathbb{S}$.
${\bf COP}$ is an infinite-dimensional problem in the sense that the dimension of the involved entries and unknowns is infinite. 

The main objective here is to show how ${\bf COP}$ can be solved up to an arbitrarily small error. This objective is achieved through a three-step approach. In the first step, we establish the concept of a truncated TBLMI. The second step outlines the process of amalgamating truncation and banded approximation operations, thereby transforming the problem into a finite-dimensional form. In the third and final step, we prove that the solution to ${\bf COP}$ can be reliably obtained up to an arbitrarily small error by solving a finite optimization problem.

\subsection{Truncation of a TBLMI}
Consider a $n \times m$ matrix  function $A$ ($\in L^2_{loc}(\mathbb{R},\mathbb{R}^{n\times m}$) and its TB transformation
\begin{equation} 
\mathcal{A} :=\mathcal{T}(A)= \begin{bmatrix}
\mathcal{A}_{11} & \ldots  & \mathcal{A}_{1m} \\ 
\vdots&&\vdots \\ 
\mathcal{A}_{n1} &\ldots & \mathcal{A}_{nm}
\end{bmatrix} 
\end{equation}
with $\mathcal{A}_{ij}:=\mathcal{T}(a_{ij})$, $i=1,\cdots, n,\ j=1,\cdots, m$.
\begin{definition}

The $n\times m$ Hankel block matrices $\mathcal{H}(A^+ )$, $\mathcal{H}(A^- )$ associated to $\mathcal{A}$ are given by: 
\begin{align}
	\mathcal{H}(A^\pm) &:= \left[\begin{array}{ccc}\mathcal{H}(A_{11}^\pm) & \cdots & \mathcal{H}(A_{1m}^\pm) \\\vdots &  &  \vdots\\ \mathcal{H}(A_{n1}^\pm) &\cdots  & \mathcal{H}(A_{nm}^\pm)  \end{array}\right].\label{hankel}
\end{align}
where for $i=1,\cdots, n,\ j=1,\cdots, m$
\begin{align*}
	\mathcal{H}(A_{ij}^+) &:= \left[\begin{array}{ccc}a_{ij,1} & a_{ij,2} & \cdots \\a_{ij,2} & \ddots &  \\\vdots &  & \ddots \end{array}\right],\\
	 \mathcal{H}(A_{ij}^-) &:= \left[\begin{array}{ccc}a_{ij,-1} & a_{ij,-2} & \cdots \\a_{ij,-2} & \ddots &  \\\vdots &  & \ddots \end{array}\right].
\end{align*}
with $a_{ij,k}$, $k\in\mathbb{Z}$ the phasor sequence of $a_{ij}$. 
\end{definition}

\begin{definition}\label{proj} Consider $n\times m$ infinite-dimensional {TB} matrices $\mathcal{A}:=\mathcal{T}(A)$. The truncation operator $\Pi_r$ at order $r$ is defined  by:
\begin{align}
&\Pi_r(\mathcal{A}) := \begin{bmatrix}
\Pi_r(\mathcal{A}_{11}) & \ldots  & \Pi_r(\mathcal{A}_{1m}) \\ 
\vdots&&\vdots \\ 
\Pi_r(\mathcal{A}_{n1}) &\ldots & \Pi_r(\mathcal{A}_{nm})
\end{bmatrix} 
\label{proj2}
\end{align}
where $\Pi_r(\mathcal{A}_{ij})$ refers to the $(2r+1) \times (2r+1)$ principal submatrix of $\mathcal{A}_{ij}$.
For the associated Hankel matrices as defined by \eqref{hankel}, we denote by 
$\mathcal{H}_{(r_1,r_2)}(\cdot)$ the finite dimensional matrix obtained by selecting the first $(2r_1+1)$ rows and $(2r_2+1)$ columns in each of its $n\times m$ infinite dimensional blocks.
\end{definition}

Deriving algebraic rules for infinite-dimensional TB matrix functions of compatible size is straightforward, given that both the sum and the product of two TB matrices result in another TB matrix.  In finite dimension, the situation is more complicated as outlined in the following result, proved in \cite{riedinger2022solving}, which explains why the product of two finite dimensional Toeplitz matrices is not a Toeplitz matrix.
\begin{theorem}Consider two  infinite dimensional TB matrix functions $\mathcal{A}$ and $\mathcal{B}$.
The following relations are fulfilled:\\
If $\mathcal{A}$ and $\mathcal{B}$ are both formed by $n\times m$ blocks:
\begin{align}
&\Pi_r(\mathcal{A}+\mathcal{B})=\Pi_r(\mathcal{A}) +\Pi_r(\mathcal{B})  \label{pro1}
\end{align}
If $\mathcal{A}$ and $\mathcal{B}$ are formed respectively by $n\times p$ and $p\times m$ blocks: 
\begin{align}
&\Pi_r(\mathcal{AB})= \Pi_r(\mathcal{A})\Pi_r(\mathcal{B})+\mathcal{H}_{(r,\eta)}(A^+)\mathcal{H}_{(\eta,r)}(B^-) \nonumber  \\&\qquad \qquad \quad+\mathcal{J}_{n,r} \mathcal{H}_{(r,\eta)}(A^-)\mathcal{H}_{(\eta,r)}(B^+)\mathcal{J}_{m,r}\label{pro}
\end{align}
where \footnote{recall from notations that $J_r$ is the $(2r+1) \times (2r+1)$ flip matrix having 1 on the anti-diagonal and zeros elsewhere.} $\mathcal{J}_{n,r}:=Id_n\otimes J_r$  and $\eta:=\min\{\xi\in\mathbb{Z}^+\cup \{+\infty\}:\xi\geq\frac{1}{2}( \min {(d^oA,d^oB)}-1)\}$ {with $d^o A$ the largest non vanishing Fourier coefficient of $A$ (the largest harmonic)}.
\end{theorem}

\begin{figure}\begin{center}
		\includegraphics[scale=0.2]{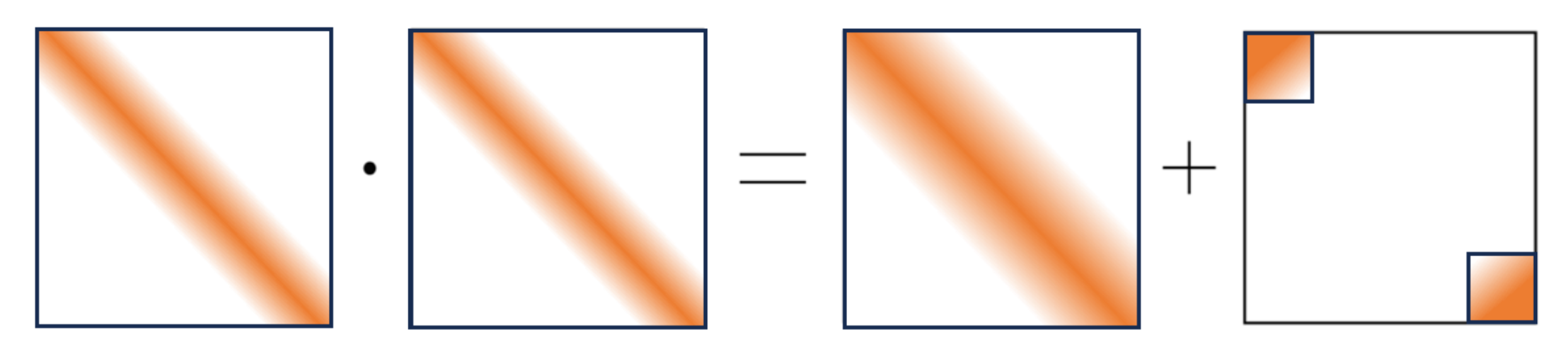}
		\caption{Multiplication of two finite dimensional banded Toeplitz matrices}\label{fig1}
	\end{center}
\end{figure}
An illustration of the above theorem is given in Fig.~\ref{fig1} for $n=p=m=1$ when $d^oA$ and
$d^oB$ are less than $r$ so that $\Pi_r(\mathcal{A})$ and $\Pi_r(\mathcal{B})$
are banded matrices.
In this case, the matrices $E^+:=\mathcal{H}_{(r,\eta)}(A^+)\mathcal{H}_{(\eta,r)}(B^-)$ and $E^-:= J_r\mathcal{H}_{(r,\eta)}(A^-)\mathcal{H}_{(\eta,r)}(B^+) J_r$ have disjoint supports located in the upper leftmost corner and in the lower rightmost corner, respectively. As a consequence,
$\Pi_r(\mathcal{A})\Pi_r(\mathcal{B})$ can be represented as the sum of $\Pi_r(\mathcal{C})$ and two correcting terms $E^+$ and $E^-$.

We are now prepared to provide a precise definition for an $r-$truncation of a generic TBLMI \eqref{LMg}.

\begin{definition}
For a given $r>0$, the $r-$truncated TBLMI of \eqref{LMg} is {defined by}:
 \begin{equation}\Pi_r(\mathcal{L}(\mathcal{P};\mathcal{A}_{s}, s\in \mathbb{S}))<0\label{lmi_trunc}\end{equation}
 \end{definition}
For example, the $r-$truncated TBLMI associated to \eqref{lyap}~is: 
\begin{align}&{\Pi_r((\mathcal{A}-\mathcal{N})^*)\Pi_r(\mathcal{P})+\Pi_r(\mathcal{P})\Pi_r(\mathcal{A}-\mathcal{N})}\nonumber\\
& +\mathcal{H}_{(r,\eta)}(A^{*+})\mathcal{H}_{(\eta,r)}(P^-)+\mathcal{H}_{(r,\eta)}(P^{+})\mathcal{H}_{(\eta,r)}(A^{-})\nonumber\\
 &+\mathcal{J}_{n,r} (\mathcal{H}_{(r,\eta)}(A^{*-})\mathcal{H}_{(\eta,r)}(P^+)\nonumber\\&+ \mathcal{H}_{(r,\eta)}(P^-)\mathcal{H}_{(\eta,r)}(A^+))\mathcal{J}_{n,r}<0\nonumber
\end{align}
with $\eta:= \min\{\xi\in\mathbb{Z}^+\cup \{+\infty\}:\xi\geq\frac{1}{2}( \min {(d^oA,d^oP)}-1)\}$.

The following result asserts that if the infinite-dimensional TBLMI (\ref{LMg}) is feasible, then a solution to the truncated TBLMI (\ref{lmi_trunc}) can always be found for any order $r$.
\begin{theorem}\label{sol_trunc}If $\mathcal{P}$ solves the infinite-dimensional TBLMI (\ref{LMg}) then $\mathcal{P}$ 
solves the $r$-truncated TBLMI (\ref{lmi_trunc}) at any order~$r$.
\end{theorem}
\begin{proof} Consider a solution $\mathcal{P}$ to \eqref{LMg} then for any $r>0$, the principal submatrix $\Pi_r(\mathcal{L}(\mathcal{P};{\mathcal{A}_s,s\in \mathbb{S}}))$ {of $\mathcal{L}(\mathcal{P};{\mathcal{A}_s,s\in \mathbb{S}})$} is necessarily negative definite.  Hence, one of the solutions to the resulting $r-$truncated TBLMI is, in fact, $\mathcal{P}$ itself.
\end{proof}

Due to Properties \eqref{pro1} and \eqref{pro}, it is possible to explicitly expand $\Pi_r(\mathcal{L}(\mathcal{P};{\mathcal{A}_s,s\in \mathbb{S}}))$ without resorting to any approximation. This expansion results in a finite-dimensional problem, provided that all the values obtained for $\eta$ in \eqref{pro} are finite. However, if some values of $\eta$ turn out to be infinite, then \eqref{lmi_trunc} will contain terms involving infinite-dimensional Hankel matrices. In the following, we will demonstrate how to transform this infinite-dimensional problem into a finite one by utilizing a banded approximation for the matrix entries.

\subsection{Truncated and banded approximation of {TBLMI}}
 {The aim of this part is to show that \eqref{LMg} can be approximated by a banded version (see \eqref{blmi}) whose $r$-truncation (see \eqref{lmi_trunc3}) is now tractable numerically since only a finite number of unknowns must be taken into account.}
%\begin{theorem}\label{conv_l2}Assume that $\mathcal{A}$ is a bounded operator on $\ell^2$. The operator $\mathcal{A}_{b(p)}$ converges to $\mathcal{A}$ in $\ell^2$-operator norm i.e.
%$$\lim_{p\rightarrow +\infty}\|\mathcal{A}-\mathcal{A}_{b(p)}\|_{\ell^2}=0$$
%\end{theorem}
%\begin{proof}
%
%\end{proof}
%This result allows to replace any TBLMI by its banded version as stated in the following theorem.
\begin{theorem}\label{tt1} Assume that $\mathcal{A}$ is a bounded operator on $\ell^2$ and denote by $\mathcal{A}_{b(p)}$ 
 its $p-$banded version obtained by deleting all its phasors of order higher than $p$. The following results hold true:
\begin{enumerate}
\item The operator $\mathcal{A}_{b(p)}$ converges to $\mathcal{A}$ in $\ell^2$-operator norm~i.e.
$$\lim_{p\rightarrow +\infty}\|\mathcal{A}-\mathcal{A}_{b(p)}\|_{\ell^2}=0$$

\item If ${\mathcal{P}}$ is a solution to \eqref{LMg} then there exists $p_0$ such that for any $p\geq p_0$, 
%$$\|\mathcal{L}(\mathcal{P};{\mathcal{A}_{s}, s\in \mathbb{S}})-\mathcal{L}(\mathcal{P};{\mathcal{A}_{s_{b(p)}}, s\in \mathbb{S}})\|_{\ell^2}<\epsilon$$
%which implies that $\mathcal{P}$ satisfies the $p-$banded TBLMI:
\begin{equation}\mathcal{L}({\mathcal{P}};{\mathcal{A}_{s_{b(p)}}, s\in \mathbb{S}})<0 \label{blmi}\end{equation} %for sufficiently small $\epsilon$.
\item For given $p$ and $r$, the $r-$truncated and $p-$banded TBLMI: 
\begin{equation} 
\Pi_r(\mathcal{L}(\mathcal{P};{\mathcal{A}_{s_{b(p)}}, s\in \mathbb{S}}))<0 \label{lmi_trunc3}.
\end{equation}
involves a finite number of unknown phasors of $\mathcal{P}$.  
 \end{enumerate}
\end{theorem}

\begin{proof}
{Let us show the first assertion. As $\|\mathcal{A}\|_{\ell^2}=\|A\|_{L^\infty}$ where $\mathcal{A}:=\mathcal{T}(A)$ (see Theorem \ref{borne}) and using the Fourier series of $A$: $$A(t)=\sum_{k\in\mathbb{Z}} A_ke^{ \textsf{j}\omega kt} \ a.e.,$$} we can write: 
\begin{align}
\|\mathcal{A}-\mathcal{A}_{b(p)}\|_{\ell^2}& =\|A-A_{b(p)}\|_{L^\infty}=\|\sum_{|k|>p} A_ke^{ \textsf{j}\omega kt} \|_{L^\infty}\label{e1}
\end{align}
As by assumption there exists a constant $C_1$ such that
\begin{align*}
\|\mathcal{A}\|_{\ell^2}=\|A\|_{L^\infty}
&= \|\sum_{k\in \mathbb{Z}} A_ke^{ \textsf{j}\omega kt} \|_{L^\infty}<C_1
\end{align*}
the series $\sum_{k\in \mathbb{Z}} A_ke^{ \textsf{j}\omega kt}$ converges almost everywhere and $\lim_{p\rightarrow +\infty}\sum_{|k|>p} A_ke^{ \textsf{j}\omega kt}=0\ a.e.$ 
Taking the limit w.r.t. $p$ in \eqref{e1} leads to the result. \\
{Now for assertion 2),} as the entries $\mathcal{A}_{s}, s\in \mathbb{S}$ are assumed bounded on $\ell^2$, the only term of the TBLMI not bounded on $\ell^2$ is $\mathcal{N}$. Fortunately, as $\mathcal{N}$ is diagonal, $\mathcal{N}=\mathcal{N}_{{b(p)}}$ for any $p\geq 0$ and thus $\mathcal{N}$ does not play any role.
If ${\mathcal{P}}$ is a solution to \eqref{LMg}, then by assumption $\mathcal{L}({\mathcal{P}};{\mathcal{A}_{s}, s\in \mathbb{S}})$ must be a bounded operator on $\ell^2$. By continuity property of LMIs with respect to their entries, there exists a constant $C_2$ depending of ${\mathcal{P}}$ and {$\mathcal{A}_s,$ $s\in \mathbb{S}$}, such that for any $p>0$
\begin{align*}\|\mathcal{L}({\mathcal{P}};{\mathcal{A}_{s}, s\in \mathbb{S}})-\mathcal{L}({\mathcal{P}};&{\mathcal{A}_{s_{b(p)}}, s\in \mathbb{S}})\|_{\ell^2}\\&\leq C_2 \sum_{s\in {\mathbb{S}}} \|\mathcal{A}_{s}-\mathcal{A}_{s_{b(p)}}\|_{\ell^2}\end{align*}
From the first assertion, we conclude that for any $\epsilon>0$, there exists $p_0$ such that for $p\geq p_0$, 
$$\|\mathcal{L}({\mathcal{P}};{\mathcal{A}_{s}, s\in \mathbb{S}})-\mathcal{L}({\mathcal{P}};{\mathcal{A}_{s_{b(p)}}, s\in \mathbb{S}})\|_{\ell^2}<\epsilon$$
and relation \eqref{blmi} follows for sufficiently small $\epsilon$.\\
Finally, to show the last assertion, as all {$\mathcal{A}_{s_{b(p)}}$, $s\in \mathbb{S}$}, in $\mathcal{L}(\mathcal{P};{\mathcal{A}_{s_{b(p)}}, s\in \mathbb{S}})<0$ are banded, only the unknown $\mathcal{P}$ is possibly not banded.
As the product of infinite dimensional banded {TB} operators is a banded {TB} operator \footnote{This is not true in finite dimension}, the terms in the TBLMI involving operator $\mathcal{P}$ have the generic form: 
$\mathcal{U}\mathcal{P}\mathcal{V}$ where $\mathcal{U}$ and $\mathcal{V}$ are polynomial functions of banded entries ${\mathcal{A}_{s_{b(p)}}, s\in \mathbb{S}}$, and are therefore banded. 
Applying $\Pi_r$ on $\mathcal{L}(\mathcal{P};{\mathcal{A}_{s_{b(p)}}, s\in \mathbb{S}})$ leads to compute $\Pi_r(\mathcal{U}\mathcal{P}\mathcal{V})$. Using \eqref{pro}, we have:
\begin{align}
\Pi_r(\mathcal{U}\mathcal{P}\mathcal{V})&= \Pi_r(\mathcal{U})\Pi_r(\mathcal{PV})\label{ee1} \\
&+\mathcal{H}_{(r,\eta_1)}(U^+)\mathcal{H}_{(\eta_1,r)}((PV)^-) \nonumber \\
&+\mathcal{J}_{n,r} \mathcal{H}_{(r,\eta_1)}(U^-)\mathcal{H}_{(\eta_1,r)}((PV)^+)\mathcal{J}_{n,r}\nonumber 
\end{align}
where $\eta_1$ is the first integer greater than $\frac{1}{2} d^o U$ and where $\Pi_r(\mathcal{PV)}$ is determined using \eqref{pro} with $\eta$ the first integer greater than $\frac{1}{2} d^o V$. Noticing that the coefficient of the {highest} degree invoked in the Hankel matrix $\mathcal{H}_{(r,\eta)}(\cdot)$ is of degree $2(r+\eta)+1$, it is straightforward to check that only a finite number of phasors of $\mathcal{P}$ are necessary to compute both $\Pi_r(\mathcal{PV)}$ and \eqref{ee1} and thus the result follows. %and \eqref{ee2}, the result is established. 
\end{proof}
We want to emphasize that solving the LMI presented in point 3) provides us with the ability to explicitly calculate the unknown phasors of $\mathcal{P}$ up to a specified order. Given that the phasor sequence belongs to $\ell^2$, it is evident that this sequence must diminish for higher-order phasors. In essence, this implies that we can anticipate achieving a precise solution as we increase the values of both $p$ and $r$ significantly. In the subsequent section, we will delineate the precise steps to determine the solution to ${\bf COP}$ up to a small arbitrary error.

\vspace{-.1cm}
\subsection{Solving {\bf COP} up to an arbitrary error}
{
We are now ready to prove the main result of this section. To this end, we define three subproblems: the $p-$banded problem ${\bf COP_{p}}$, the fully banded problem ${\bf COP_{p,q}} $ and the $r-$truncated, fully banded ${\bf COP_{p,q,r}}$. The main result states that solving ${\bf COP_{p,q,r}}$ is a consistent scheme allowing to approximate the solution to ${\bf COP}$.}\\ 
For a given $p>0$, consider the $p-$banded problem: 
\begin{align*}
{\bf COP_{p}:} &\min_{\mathcal{P}^*=\mathcal{P}>0} tr_0(\mathcal{P}) \text{ subject to:}\\
&\mathcal{L}(\mathcal{P};{\mathcal{A}_{s_{b(p)}}, s\in \mathbb{S}})\leq 0.\nonumber
\end{align*}
For a given $q>0$, the fully banded problem is: 
\begin{align*}&{\bf COP_{p,q}:} \min_{\mathcal{P}^*=\mathcal{P}>0} tr_0(\mathcal{P}) \text{ subject to:} \\
&\mathcal{L}(\mathcal{P};{\mathcal{A}_{s_{b(p)}}, s\in \mathbb{S}})\leq 0, \quad%\nonumber \\
P_{ij,k}=0, |k|>q,\ i,j=1,\cdots,n\nonumber
\end{align*}
{where for a given $k\in\mathbb{Z}$, $P_{ij,k}=0$ refers to the $k$th-phasors of the $(i,j)$th block of $\mathcal{P}$.}\\
For a given $r>0$, the $r-$truncated, fully banded optimization problem is: 
\begin{align*}
&{\bf COP_{p,q,r}:} \min_{\mathcal{P}^*=\mathcal{P}} tr_0(\mathcal{P}) \text{ subject to:}\quad \Pi_r(\mathcal{P})>0,\\
&\Pi_r(\mathcal{L}(\mathcal{P};{\mathcal{A}_{s_{b(p)}}, s\in \mathbb{S}}))\leq 0,\ %\nonumber \\
P_{ij,k}=0, |k|>q,\ i,j=1,\cdots,n.\nonumber
\end{align*}
{It is important to note that evaluating $tr_0(\mathcal{P})$ necessitates the computation of a finite number of phasors, making ${\bf COP_{p,q,r}}$ a problem of finite dimensionality.}

We assume that all these convex optimization problems are feasible and that the optimal solution is unique, bounded on $\ell^2$ and continuous with respect to the entries $\mathcal{A}_s$, $s\in \mathbb{S}$. 
%{\begin{assumption}\label{as2}
%For given $p,q,r>0$, any sequence $\mathcal{P}_k$, $k=0,1,\cdots$ such that for any $k$, $\mathcal{P}_k$ solves ${\bf COP_{p,q,r}}$ and $\lim_{k\rightarrow +\infty} \|\mathcal{P}_k\|_{\ell^2}=+\infty$ leads to $\lim_{k\rightarrow +\infty} tr_0(\mathcal{P}_k)=+\infty$.\end{assumption}
Given $p,q$ and $r$, we denote by $\hat{\mathcal{P}}$, $\hat{\mathcal{P}}_{p}$, $\hat{\mathcal{P}}_{p,q}$ and $\hat{\mathcal{P}}_{p,q,r}$, the solution to ${\bf COP}$, ${\bf COP_{p}}$, ${\bf COP_{p,q}}$ and ${\bf COP_{p,q,r}}$ respectively.
The next theorem states that solving ${\bf COP_{p,q,r}}$ is a consistent scheme allowing to approximate the solution to ${\bf COP}$.
\begin{theorem}\label{prop} For any $\epsilon>0$, there exist $p$, $q$ and $r_0$ such that  for any $r>r_0$:
\begin{align}
& \|\hat{\mathcal{P}}_{p,q,r}-\hat{\mathcal{P}}\|_{\ell^2}=\|\hat{{P}}_{p,q,r}-\hat{{P}}\|_{L^\infty}< \epsilon
\end{align}
where $\hat{\mathcal{P}}_{p,q,r}=\mathcal{T}(\hat{{P}}_{p,q,r})$ and $\hat{\mathcal{P}}:=\mathcal{T}(\hat{{P}})$.
\end{theorem}
\begin{proof}
{Let us show that the following three limits hold:
\begin{align*}
&\lim_{p \rightarrow +\infty} \|\hat{\mathcal{P}}_{p}-\hat{\mathcal{P}}\|_{\ell^2}=0, \quad
 \lim_{q \rightarrow +\infty} \|\hat{\mathcal{P}}_{p,q}-\hat{\mathcal{P}}_p\|_{\ell^2}=0,\ p>0\\
& \text{and }\lim_{r \rightarrow +\infty} \|\hat{\mathcal{P}}_{p,q,r}-\hat{\mathcal{P}}_{p,q}\|_{\ell^2}=0, p,q>0.
\end{align*}
The first limit} is a direct consequence of the continuity of the optimal solution with respect to the entries {$\mathcal{A}_s$, $s\in \mathbb{S}$} and 1) in Theorem~\ref{tt1}.
%%%%%%%%%%%%%%%%%%%%%%%%%%%%%%%
To prove the second limit, for a given $p$, as for any $q$, $\hat{\mathcal{P}}_{p,q}$ is admissible for both ${\bf COP_{p,q+1}}$ and ${\bf COP_{p}}$, it follows necessarily that
\begin{equation}tr_0(\hat{\mathcal{P}}_{p,q})\geq tr_0(\hat{\mathcal{P}}_{p,q+1})\geq\cdots \geq tr_0(\hat{\mathcal{P}}_{p}) >0\label{des}\end{equation}
On the other hand, from 1) in Theorem~\ref{tt1}, it is clear that there exists $q_0$ such that for any $q>q_0$ the $q-$banded operator $\hat{\mathcal{P}}_{p_{b(q)}}$ of $\hat{\mathcal{P}}_{p}$ is positive definite. Thus, $\hat{\mathcal{P}}_{p_{b(q)}}$ is then obviously admissible for  ${\bf COP_{p,q}}$, it follows that: 
\begin{equation}
tr_0(\hat{\mathcal{P}}_{p_{b(q)}})\geq tr_0(\hat{\mathcal{P}}_{p,q})\geq tr_0(\hat{\mathcal{P}}_{p})\label{tg}\end{equation}
Since $\|\hat{\mathcal{P}}_{p_{b(q)}}-\hat{\mathcal{P}}_{p}\|_{\ell^2}\rightarrow0$ when $q\rightarrow+\infty$, taking the limit w.r.t. $q$  in \eqref{tg} leads to:
\begin{equation}\lim_{q\rightarrow+\infty}tr_0(\hat{\mathcal{P}}_{p,q})= tr_0(\hat{\mathcal{P}}_{p}).\label{tg2}\end{equation}
 Now, let us show that we also have: $\lim_{q\rightarrow+\infty}\hat{\mathcal{P}}_{p,q}= \hat{\mathcal{P}}_{p}$ on $\ell^2$.
As $ \hat{\mathcal{P}}_{p,q}$ is TB, Hermitian, positive definite and bounded on $\ell^2$, there exists a bounded operator on $\ell^2$, $\mathcal{Z}_{p,q}$ such that the following decomposition holds:$$ \hat{\mathcal{P}}_{p,q}=\mathcal{Z}^*_{p,q}\mathcal{Z}_{p,q} \text{ for any }p,q.$$
Moreover as $\mathcal{Z}_{p,q}$ is a constant matrix function, it belongs trivially in $H$ (see Def.~\ref{H}) and there exists a representative $Z_{p,q}\in L^\infty([0\ T])$ (see Theorem \ref{borne}) such that
$\mathcal{Z}_{p,q}=\mathcal{T}(Z_{p,q})$.
Using similar arguments, $\hat{\mathcal{P}}_{p}=\mathcal{Z}^*_{p}\mathcal{Z}_{p}$ with $\mathcal{Z}_{p}=\mathcal{T}(Z_{p})$ and $Z_{p}\in L^\infty([0\ T])$. 
Therefore, Def. 2 implies: $$tr_0( \hat{\mathcal{P}}_{p,q})=tr_0(\mathcal{Z}^*_{p,q}\mathcal{Z}_{p,q})=<Z_{p,q},Z_{p,q}>$$
and from \eqref{tg2}, it can be concluded that 
\begin{equation}\lim_{q\rightarrow+\infty}<Z_{p,q},Z_{p,q}>=<Z_{p},Z_{p}>.\label{nc}
\end{equation}
Moreover as the sequence $Z_{p,q}$ indexed by $q$ is bounded (see \eqref{des}), there exists a subsequence that converges weakly on $L^\infty$ and 
Eq. \eqref{nc} implies that it also converges strongly and necessarily to $Z_{p}$ by uniqueness of solution. Finally, the uniqueness of the solution implies that the whole sequence converges to  $Z_{p}$. It follows %since $\mathcal{C}$ is assumed to be invertible 
that  $\lim_{q\rightarrow+\infty} \hat {\mathcal{P}}_{q,p}=\hat{\mathcal{P}}_{p}=\mathcal{Z}^*_{p}\mathcal{Z}_{p}$ on $\ell^2$.
Finally to prove the third limit:  for any $r$ and following similar steps as the proof of Theorem~\ref{sol_trunc}, as $\hat{\mathcal{P}}_{p,q}$ is admissible for Problem ${\bf COP_{p,q,r}} $ and $\hat{\mathcal{P}}_{p,q,r+1}$ is admissible for Problem ${\bf COP_{p,q,r}}$ , it follows that:{
\begin{equation}
tr_0(\hat{\mathcal{P}}_{p,q,r})\leq tr_0(\hat{\mathcal{P}}_{p,q,r+1})\leq \cdots \leq tr_0(\hat{\mathcal{P}}_{p,q})\label{bound2}\end{equation}}
which proves that the sequence $tr_0(\hat{\mathcal{P}}_{p,q,r})$ indexed by $r$ is an increasing and bounded real sequence, and thus a converging sequence. Moreover, for any $r$, as there exists $\hat{\mathcal{Z}}_{p,q,r}$ s.t.  $\hat{\mathcal{P}}_{p,q,r}:=\hat{\mathcal{Z}}^*_{p,q,r}\hat{\mathcal{Z}}_{p,q,r}$, Eq. \eqref{equiv} and \eqref{bound2} imply 
$$\|\hat{\mathcal{Z}}_{p,q,r}\|_{\ell^2} \leq \sqrt{n}\|\hat{\mathcal{Z}}_{p,q}\|_{\ell^2}$$
and it follows that the sequence $\|\hat{\mathcal{P}}_{p,q,r}\|_{\ell^2}=\|\hat{\mathcal{Z}}_{p,q,r}\|^2_{\ell^2}$ indexed by $r$ is necessarily bounded on $\ell^2$.
Therefore, for any $r$, the phasors of $\hat{\mathcal{P}}_{p,q,r}$ are bounded and belong to a finite dimensional subspace of $\ell^2$ (thanks to the constraints $P_{ij,k}=0 \text{ for } |k|>q$). 
By compactness, there exists a subsequence that converges on this finite subspace of $\ell^2$ and the uniqueness of the solution implies that the whole sequence converges necessarily to $\hat{\mathcal{P}}_{p,q}$.\\
{The final result follows since $\forall\epsilon>0$, $\exists p_0$, $\forall p>p_0$, $\exists q_0(p)$, $\forall q>q_0$, $\exists r_0(p,q)$ such that:
$ \|\hat{\mathcal{P}}_{p}-\hat{\mathcal{P}}\|_{\ell^2}\leq \frac{\epsilon}{3}$, $\|\hat{\mathcal{P}}_{p,q}-\hat{\mathcal{P}}_p\|_{\ell^2}\leq \frac{\epsilon}{3}$ and $\forall r>r_0$, $\|\hat{\mathcal{P}}_{p,q,r}-\hat{\mathcal{P}}_{p,q}\|_{\ell^2}\leq \frac{\epsilon}{3}$
and thus it follows that 
\begin{align*}
\|\hat{\mathcal{P}}_{p,q,r}-\hat{\mathcal{P}}\|_{\ell^2}\leq&\|\hat{\mathcal{P}}_{p,q,r}-\hat{\mathcal{P}}_{p,q}\|_{\ell^2}\\&+\|\hat{\mathcal{P}}_{p,q}-\hat{\mathcal{P}}_p\|_{\ell^2}+ \|\hat{\mathcal{P}}_{p}-\hat{\mathcal{P}}\|_{\ell^2}\leq \epsilon\end{align*}}
The proof ends invoking Theorem~\ref{borne}.
\end{proof}

\section{Illustrative examples}
\subsection{$H_2$ harmonic control-LQR case}
We consider the example given in \cite{riedinger2022solving}:
\begin{align}
	\dot x=&\left(\begin{array}{cc}a_{11} (t) & a_{12} (t) \\a_{21} (t) & a_{22} (t)\end{array}\right)x+\left(\begin{array}{c}b_{11}(t) \\0\end{array}\right)u\label{ex_ltp}\end{align}
{\small\begin{align*}a_{11} (t) &:=1+\frac{4}{\pi}\sum_{k=0}^{\infty}\frac{1}{2k+1}\sin(\omega (2k+1)t),\\
	a_{12} (t) &:= 2+\frac{16}{\pi^2}\sum_{k=0}^{\infty}\frac{1}{(2k+1)^2}\cos(\omega (2k+1)t),\\
	a_{21} (t) &:= -1+\frac{2}{\pi}\sum_{k=1}^{\infty}\frac{(-1)^k}{k}\sin(\omega kt+\frac{\pi}{4}),\\
	a_{22} (t) &:= 1-2\sin(\omega t)-2\sin(3\omega t)+2\cos(3\omega t)+2\cos(5\omega t),\\
	b_{11}(t)&:=1+ 2 \cos(2\omega t)+ 4 \sin(3\omega t) \text{ with }\omega=2\pi.
\end{align*}}
Note that $a_{11}$, $a_{12}$ and $a_{21}$ are respectively square, triangular and sawtooth signals and include an offset part. The associated TB $\mathcal{A}$ matrix has an infinite number of phasors and is not banded. This system is unstable. Its equivalent harmonic LTI system \eqref{ltih} is characterized by a spectrum given by the set $\sigma:=\{\lambda+ \textsf{j}\omega k, k\in \mathbb{Z}\}$ where $\lambda \in \{1\pm \textsf{j} 1.64\}$ (see \cite{riedinger2022solving}).

We consider the LQR problem whose solution can be obtained by solving the associated 
infinite dimensional convex optimization problem \cite{willems1971least}:
\begin{align} 
 &\max_{\scriptsize \mathcal{P}=\mathcal{P}^*>0} tr_0(\mathcal{P}),\ % \sum_{i=1}^n Z_{ii,0}
\label{op}\\
&\left(\begin{array}{cc}
(\mathcal{A}-\mathcal{N})^*\mathcal{P}+\mathcal{P}(\mathcal{A}-\mathcal{N})+\mathcal{Q} & \mathcal{PB} \\
\mathcal{B}^*\mathcal{P}& \mathcal{R}
\end{array}\right)\geq 0\nonumber
\end{align}
where the trace operator is defined by \eqref{tr} and $\mathcal{Q}$ and $\mathcal{R}$ are the harmonic LQR weighting matrices. The matrix gain is given by $\mathcal{K}:=\mathcal{R}^{-1}\mathcal{B}^*\mathcal{P}$ where $\mathcal{P}$ is a {TB} matrix of infinite dimension and a bounded operator on $\ell^2$; see \cite{blin_necessary_2022} for more details. 
For comparison purpose, we also consider an equivalent formulation which consists in solving the following infinite dimensional convex optimization problem\cite{boyd1994linear}: 
\begin{align}
 &\min_{\scriptsize \mathcal{S}=\mathcal{S}^*>0} tr_0(\mathcal{W}),\label{op2}\\
&\left(\begin{array}{ccc}(\mathcal{A}-\mathcal{N})\mathcal{S}+\mathcal{S}(\mathcal{A}-\mathcal{N})^*-\mathcal{B}\mathcal{Y}-\mathcal{Y}^*\mathcal{B}^* & \star & \star \\ \mathcal{R}^{\frac{1}{2}}\mathcal{Y}& -\mathcal{I} & \star \\\mathcal{Q}^{\frac{1}{2}}\mathcal{S} & 0 & -\mathcal{I}\end{array}\right)\leq 0 \nonumber\\ 
&\left(\begin{array}{cc}\mathcal{W}& \mathcal{I} \\\mathcal{I} & \mathcal{S}\end{array}\right)>0 \nonumber
\end{align}
With this formulation, the matrix gain is given by $\mathcal{K}:=\mathcal{Y}\mathcal{S}^{-1}$.
Finally, the LQR problem can also be solved through a harmonic $H_2$ problem formulation provided by \eqref{PB_H2} with 
$\mathcal{C}_z:=[\mathcal{Q}^\frac{1}{2};0]$, $\mathcal{C}_y:=\mathcal{I}$ 
, $\mathcal{D}_{zu}:=[0;\mathcal{R}^\frac{1}{2}]$ and 
$\mathcal{D}_{zw}=\mathcal{D}_{yw}=\mathcal{D}_{yu}:=0$.

\begin{figure}[h]
	\begin{center}
		\includegraphics[width=\linewidth]{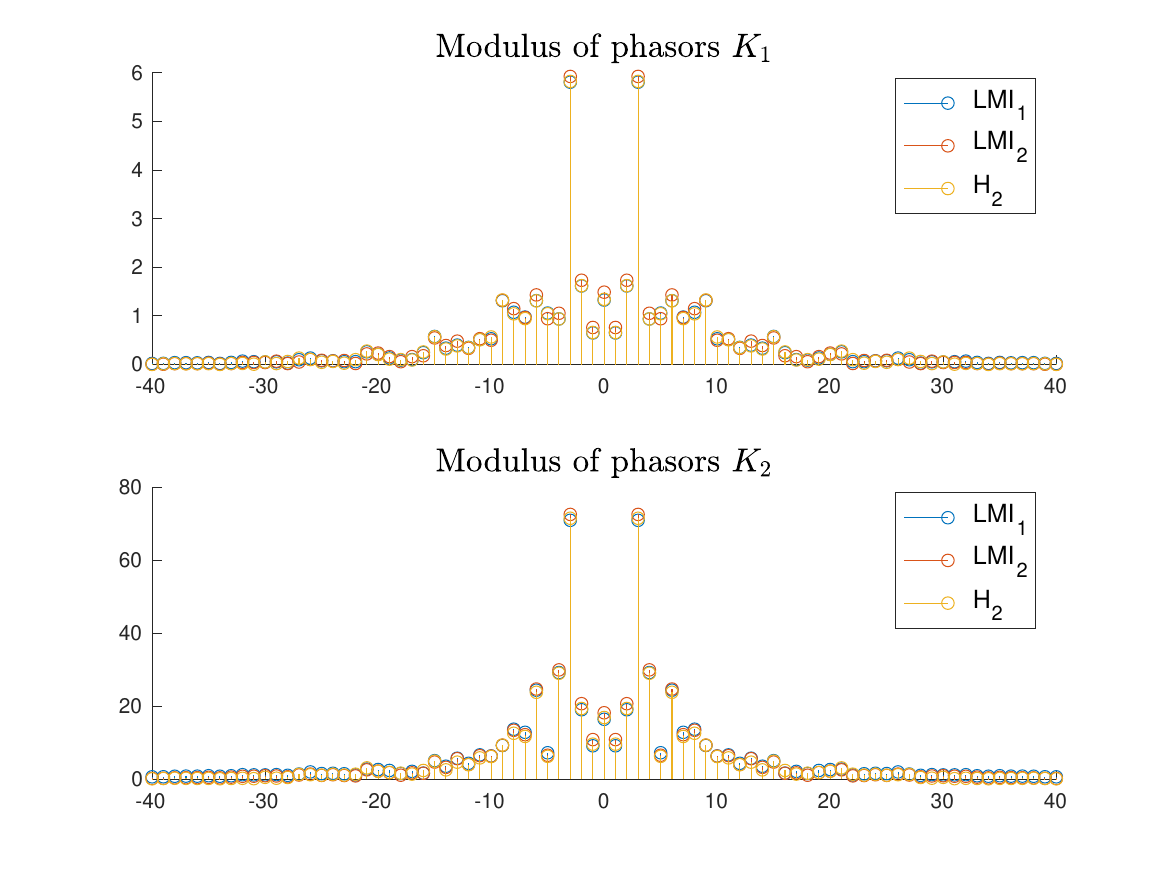}
		\caption{Moduli of gain-phasors $K=[K_1,K_2]$ for Problem \eqref{op}, \eqref{op2} \eqref{PB_H2} with $r=30$}\label{f0}
	\end{center}
\end{figure}

\begin{figure}[h]
	\begin{center}
		\includegraphics[width=\linewidth]{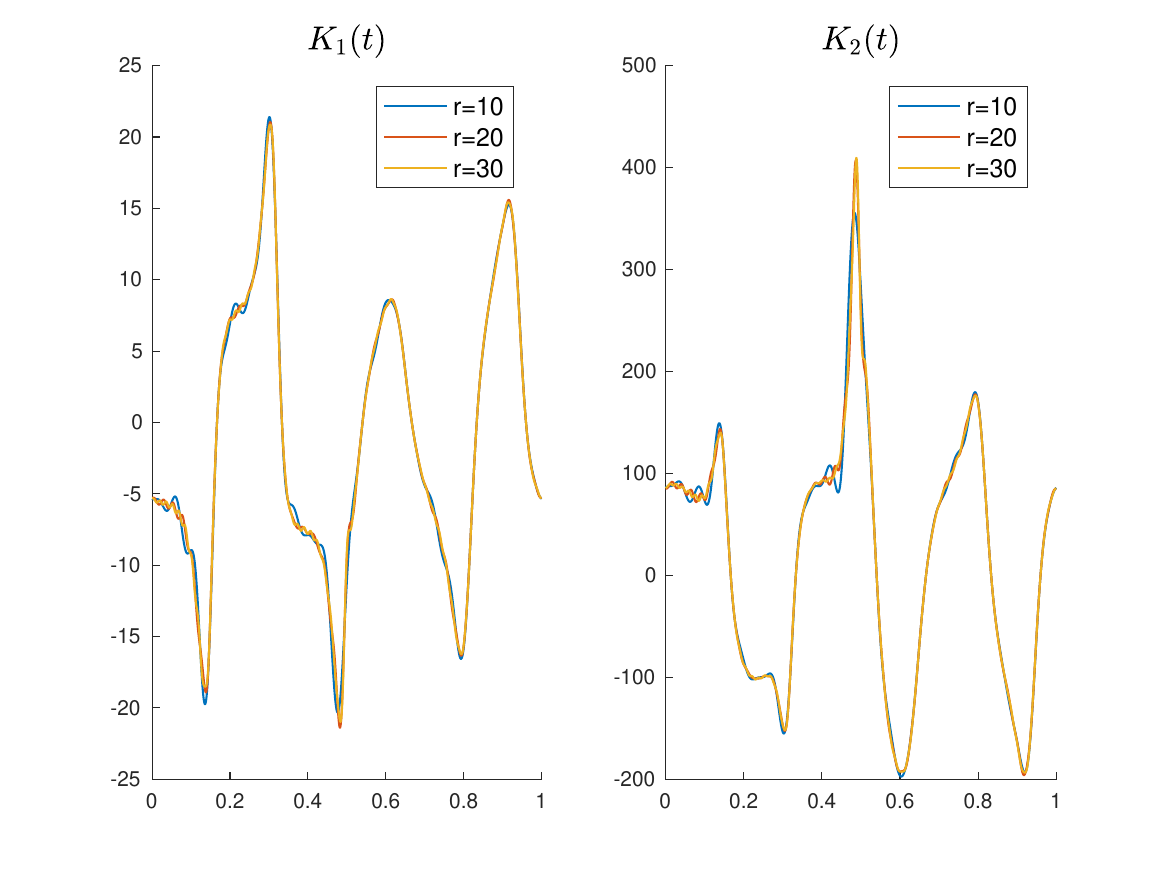}
		\caption{$T$-periodic gain $K(t)$ over a period $T=1$ for Problem \eqref{op}}\label{f3}
	\end{center}
\end{figure}

\begin{figure}[h]
	\begin{center}
		\includegraphics[width=\linewidth]{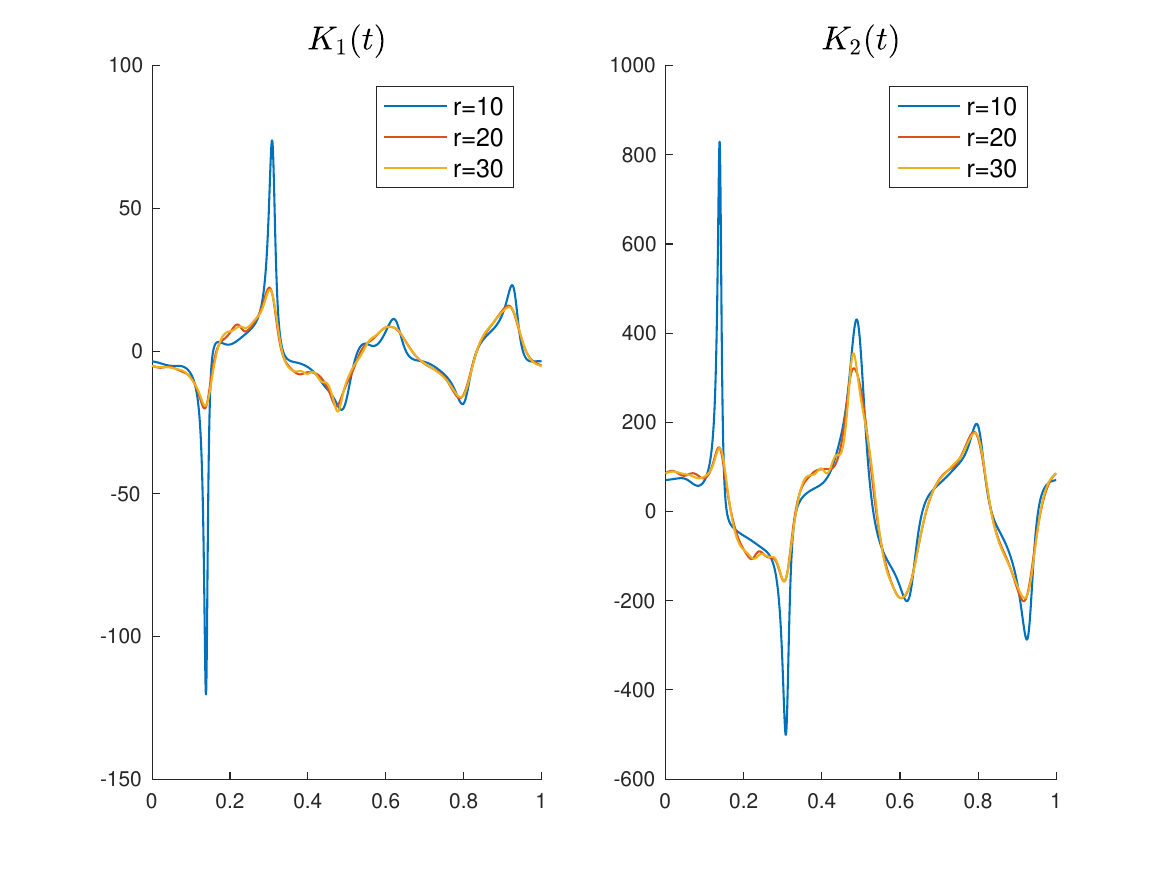}
		\caption{$T$-periodic gain $K(t)$ over a period $T=1$ for Problem \eqref{PB_H2}}\label{f1}
	\end{center}
\end{figure}

%\begin{figure}[h]
%	\begin{center}
%		\includegraphics[width=\linewidth]{Kk_H2}
%		\caption{Moduli of Phasors $K=[K_1,K_2]$ (harmonic LQ control, $H_2$ formulation)}\label{f2}
%	\end{center}
%\end{figure}
We have opted for the following choices for matrices $\mathcal{Q}$ and $\mathcal{R}$: $\mathcal{Q}:=\mathcal{T}(diag([1\ 10^4]))$ and $\mathcal{R}:=\mathcal{T}(Id_m)$. Imposing as required a TB structure to the unknown matrices, we have solved Problem ${\bf COP_{p,q,r}}$ associated with equations \eqref{op}, \eqref{op2}, and \eqref{PB_H2}. Specifically, we used $r=30$, $p=2r$, and $q=r$ for Problems \eqref{op2} and \eqref{PB_H2}, and $q=2r$ for Problem \eqref{op}.

Furthermore, the constraint of $q=r$ for Problems \eqref{op2} and \eqref{PB_H2} is based on the consideration that when $q=2r$, higher-order residual phasors emerge in $\mathcal{Y}$ and $\mathcal{S}$, which can be undesirable for the inversion of the $\mathcal{S}$ matrix. 

{With these choices and taking advantages of the TB structure, if $n=2$ denotes the dimension of the state and $m=1$ the dimension of the control, the number of scalar unknowns are respectively  $\frac{n(n+1)}{2}(2q+1)=363$ , $(n(n+1)+nm)(2q+1)=488$  and $(nm+\frac{n(n+1)}{2}+\frac{nm(nm+1)}{2})(2q+1)=671$ for problem \eqref{op}, \eqref{op2} and \eqref{PB_H2} and the computation times required to obtain solutions for these three problems are as follows: $T_{comp} = 41s, 145s \text{ and } 86s$ respectively.} In summary, the number of unknowns increases linearly with respect to the number of harmonics to be taken into account. 

For the substantial truncation order of $r=30$, Figure \ref{f0} displays the magnitudes of phasors within $\mathcal{K}=[\mathcal{K}_1,\mathcal{K}_2]$. As evident from the plot, the solutions for these three problems align closely, demonstrating a high degree of accuracy and consistency.

Now, regarding the computation times for solving \eqref{op} with different truncation orders $r=10,20,30$, we have the following respective times: $T_{comp} = 1.5s, 9.5s, 41s$. Meanwhile, when dealing with the $H_2$ problem formulation \eqref{PB_H2}, we obtain the following times: $T_{comp} = 2.7s, 22s, 94s$ for the same values of $r$. Notably, the results obtained for both optimization problems exhibit substantial similarity for $r=20,30$. However, it is worth mentioning that for $r=10$, the solution obtained by solving \eqref{PB_H2} lacks accuracy and has not been retained. For illustration purpose,  we visualize in Figure~\ref{f3} and Figure~\ref{f1} the $T-$periodic gain matrix $K(t)$ over a period $T$ for various values of $r=10,20,30$. These matrices are obtained by solving the optimization problems \eqref{op} and \eqref{PB_H2}, respectively.

Now, it is evident that the control law given by:
$$u(t):=u_{ref}(t)-K(t)(x(t)-x_{ref}(t)),$$
where $K(t)$ represents the $T-$periodic gain matrix defined as $K(t):=\sum_{k=-q}^{q} K_k e^{\textsf{j}\omega kt}$
effectively stabilizes the unstable LTP system \eqref{ex_ltp} globally and asymptotically. This stabilization occurs on any $T-$periodic trajectory characterized by $x_{ref}(t) = \mathcal{F}^{-1}(X_{ref})$ and $u_{ref}(t) = \mathcal{F}^{-1}(U_{ref})$, where the pair $(X_{ref},U_{ref})$ satisfies the harmonic equilibrium equation
\begin{align}0=(\mathcal{A}-\mathcal{N})X_{ref}+\mathcal{B}U_{ref}.\label{equi}\end{align}

To illustrate this, we visualize the closed-loop response for three $T-$periodic reference trajectories $(x_{ref},u_{ref})$ in Figure~\ref{f4}. We begin with $u_{ref}(t) := 0$ for $t < 2$. Subsequently, for $2 \leq t < 4$, we set $u_{ref}(t) := 4\cos(2\pi t)$. For $t \geq 4$, we introduce a desired steady state $X_d$ defined as $\mathcal{F}^{-1}(X_d)(t) := (\frac{1}{4}\cos(2\pi t), 0)$ and seek the nearest harmonic equilibrium. This involves solving the minimization problem $\min_{U_{ref}}|X_d-X_{ref}|^2$ subject to \eqref{equi}. It is evident from Figure~\ref{f4} that the provided state feedback enables tracking of any $T-$periodic trajectory associated with any equilibrium of \eqref{equi}.

%\begin{figure}[h]
%	\begin{center}
%		\includegraphics[width=\linewidth]{Pk_LQR}
%		\caption{Moduli of Phasors $P$ (Riccati solution)}\label{f3}
%	\end{center}
%\end{figure}
%\begin{figure}[h]
%	\begin{center}
%		\includegraphics[width=\linewidth]{Pt_LQR}
%		\caption{$T-$periodic Riccati solution $P(t)$ over a period $T=1$}\label{f4}
%	\end{center}
%\end{figure}
%
\begin{figure}[h]
	\begin{center}
		\includegraphics[width=\linewidth]{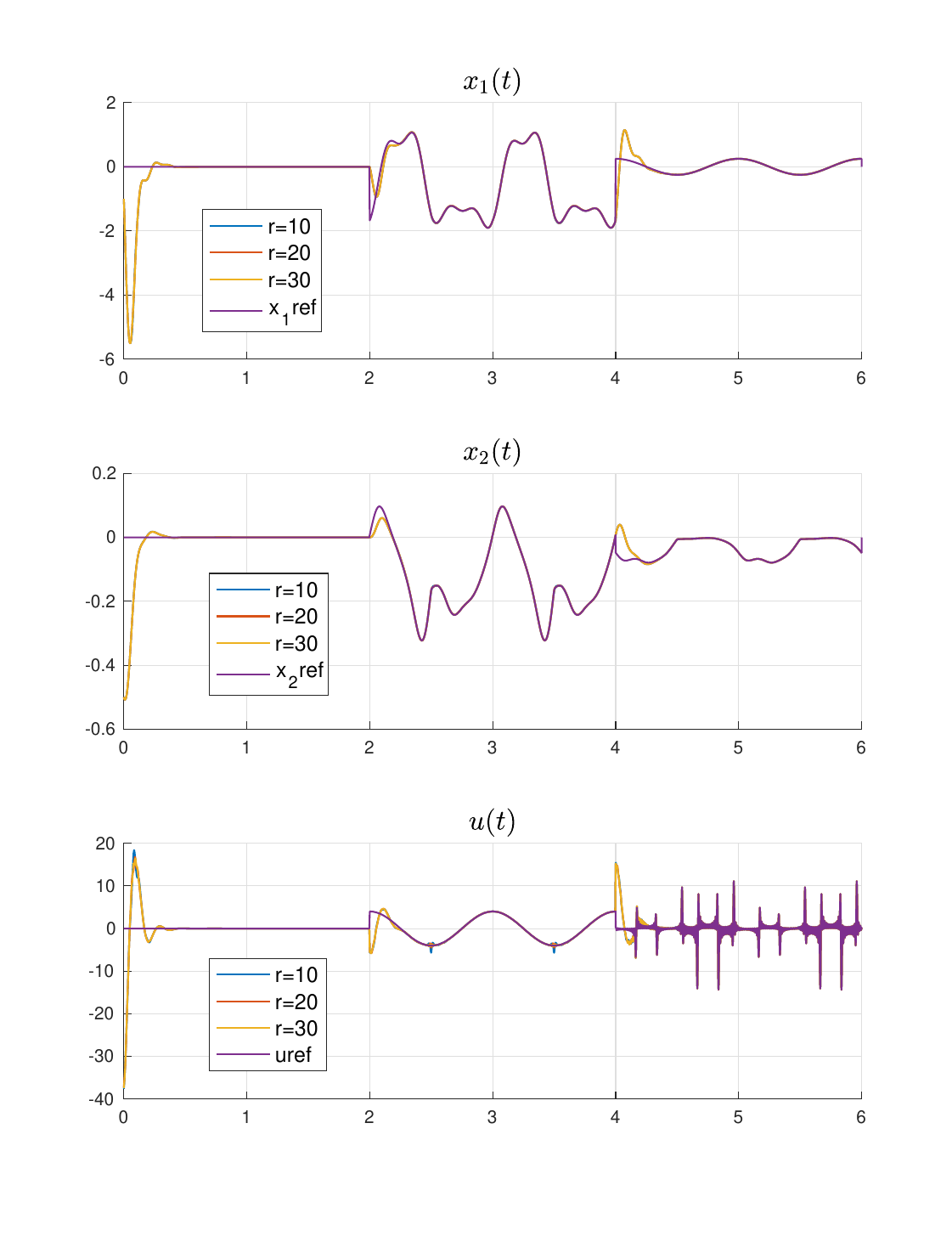}
		\caption{Closed loop response with harmonic LQ control}\label{f4}
	\end{center}
\end{figure}

%\addtolength{\textheight}{0cm} % This command serves to balance the column lengths
   % on the last page of the document manually. It shortens
   % the textheight of the last page by a suitable amount.
   % This command does not take effect until the next page
   % so it should come on the page before the last. Make
   % sure that you do not shorten the textheight too much.
%%%%%%%%%%%%%%%%%%%%%%%%%%%%%%%%%%%%%%%%%%%%%%%%%%%%%%%%%%%%%%%%%%%%%%%%%%%%%%%%

\subsection{$H_\infty$ harmonic control}\label{Hinf_example}

Now, let us consider solving the $H_\infty$ optimal full-state feedback problem for the same $T-$periodic system.
In the context of $H_\infty$ constraints, our objective is to minimize the maximum singular value of the harmonic transfer function between the input matrix $\mathcal{B}$ and the output vector $Z:=(\mathcal{Q}^{\frac{1}{2}}X,\mathcal{R}^{\frac{1}{2}}U)$. To achieve this, we set $\mathcal{B}_w:=\mathcal{B}$, $\mathcal{C}_z:=\left(\begin{array}{ccc}\mathcal{I} & 0  \\0 & 10^2\mathcal{I}\\
0&0
\end{array}\right)$, $\mathcal{D}_{zw}:=0$ and $\mathcal{D}_{zu}:=(0,0,\mathcal{I})$.
%Moreover, for a full state feedback, we have : $\mathcal{C}_y:=\mathcal{I}$, $\mathcal{D}_{yw}=\mathcal{D}_{yu}:=0$,

%The optimal solution to this $H_\infty$ harmonic optimal control problem is determined by
%$U=\mathcal{K}X$ with $\mathcal{K}:=\mathcal{S}\mathcal{P}^{-1}$ where $\mathcal{S}$ $\mathcal{P}$ are TB and bounded on %$\ell^2$ operator that solve the convex optimization problem:
%\begin{align}
%&\min_{\mathcal{P},\mathcal{S},\gamma} \gamma \text{ subject to:}\nonumber\\
%& \mathcal{P}^*=\mathcal{P}>0\label{Hinf}\\
%&\left[\begin{array}{ccc}(\mathcal{A}-\mathcal{N})\mathcal{P}+\mathcal{P}(\mathcal{A}-\mathcal{N})^*+\mathcal{BS}+ %\mathcal{S}^*\mathcal{B}^* & \star & \star\\
%\mathcal{B}_w^* & -\gamma\mathcal{I}  & \star\\
%\mathcal{C}_{z}\mathcal{P}+\mathcal{D}_{zu}\mathcal{S}& \mathcal{D}_{zw} &-\gamma %\mathcal{I}\end{array}\right]\nonumber\end{align}

Similar to the previous case, we solve Problem ${\bf COP_{p,q,r}}$ associated with \eqref{Hinf}, but this time we choose $r=10,20,30$, with $p=\frac{r}{2}$ and $q=r$. It is worth noting that in this configuration, we have intentionally set $p$ to be much smaller than in the previous example. This adjustment ensures that we obtain a practical solution even for smaller values of $r$. Figures \ref{f6} and \ref{f7} depict the magnitudes of gain-phasors and the values of $K(t)$ over a period $T$, respectively.

\begin{figure}[h]
	\begin{center}
		\includegraphics[width=\linewidth]{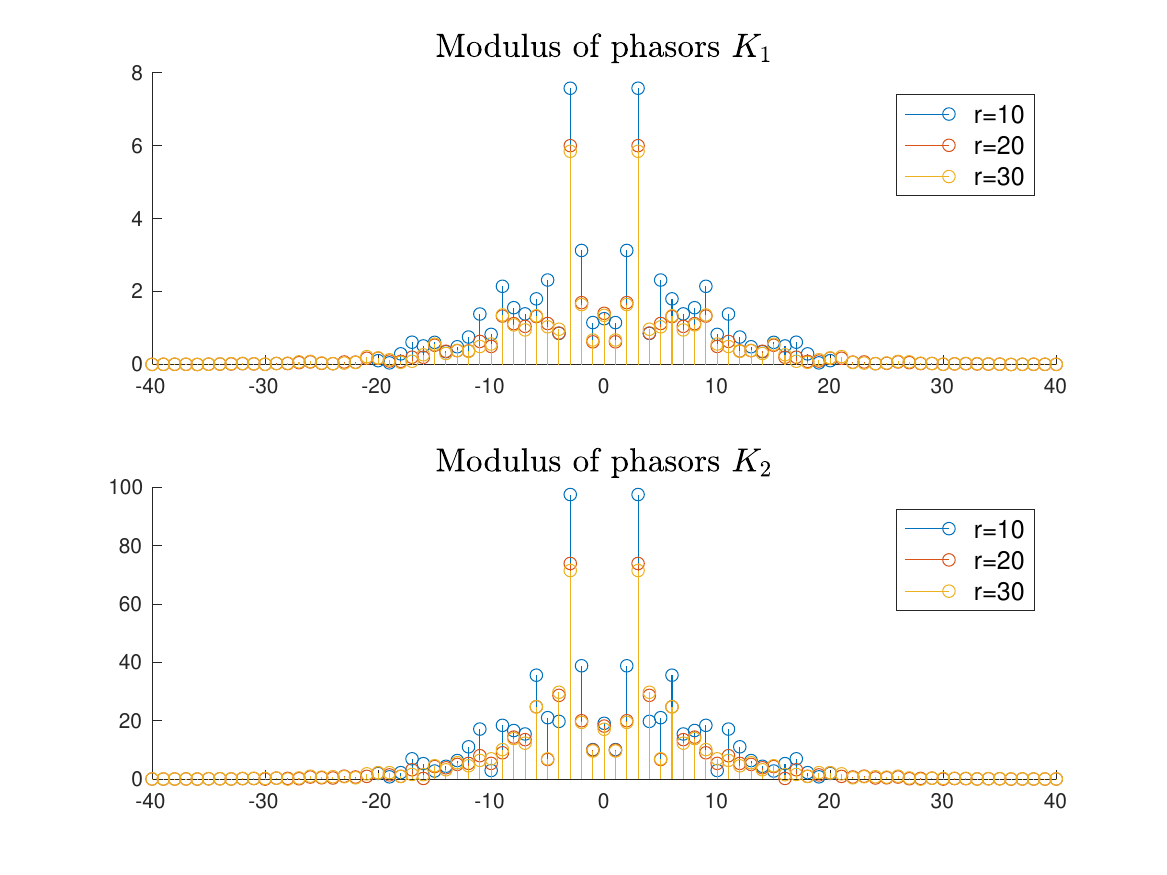}
		\caption{Moduli of gain-phasors $K=[K_1,K_2]$(Harmonic $H_\infty$)}\label{f6}
	\end{center}
\end{figure}
\begin{figure}[h]
	\begin{center}
		\includegraphics[width=\linewidth]{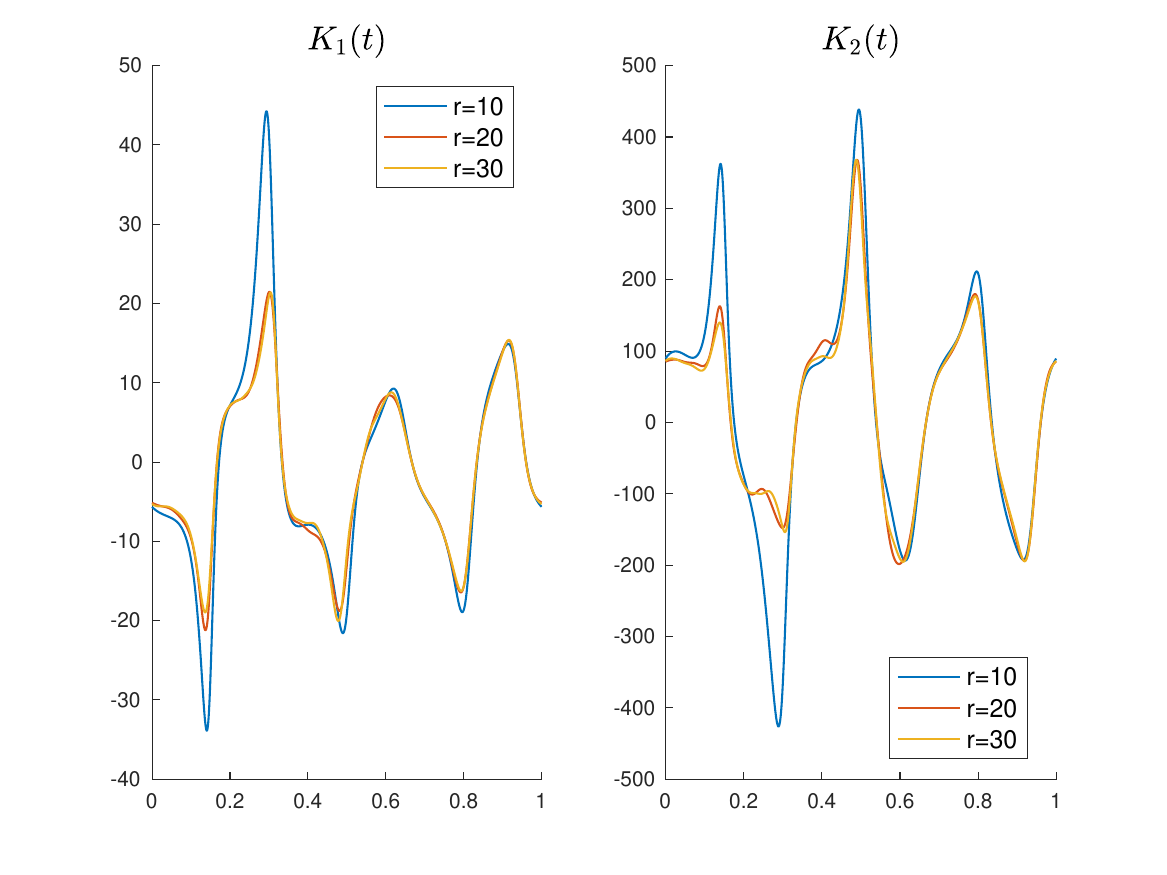}
		\caption{$T$-periodic gain $K(t)$ over a period $T=1$(Harmonic $H_\infty$)}\label{f7}
	\end{center}
\end{figure}

Finally, for the same reference trajectories, Figure \ref{f8} presents the closed-loop response. It is evident that the solution obtained for $r=10$ does not accurately track the reference trajectories. This discrepancy is primarily attributed to the fact that not enough phasors have been taken into account in the approximation of $\mathcal{A}$ by $\mathcal{A}_{b(p)}$ with $p=5$. Even though $H_{\infty}$ optimal control aims to minimize the worst-case scenario, represented here as:
$$\sup_{\|W\|_{\ell^2}=1}\|Z\|_{L^2}=\sup_{\|W\|_{\ell^2}=1}(\int_0^{+\infty}X^*\mathcal{Q}X+U^*\mathcal{R}Udt)^{ \frac{1}{2}},$$ 
the obtained gain is approximately the same as the one obtained for the LQR control problem. This similarity suggests that the LQR control meets our $H_\infty$ criterion effectively. This observation can be rationalized as follows: as detailed in \cite{bensoussan1993standard}, the optimal solution to the $H_\infty$ problem for
 $\mathcal{C}_y:=\mathcal{I}$, $\mathcal{D}_{yu}=\mathcal{D}_{yw}=\mathcal{D}_{zw}:=0$ and $\mathcal{C}_z:=\mathcal{H}$, $\mathcal{D}_{zu}:=\mathcal{I}$ satisfies the Riccati equation:
$$(\mathcal{A}-\mathcal{N})^*\mathcal{P}+\mathcal{P}(\mathcal{A}-\mathcal{N})-\mathcal{P}(\mathcal{B}\mathcal{B}^*-\frac{1}{\gamma^2}\mathcal{B}_w\mathcal{B}_w^*)\mathcal{P}+\mathcal{H}^*\mathcal{H}=0$$
In our case, as $\mathcal{H}:=\mathcal{Q}^\frac{1}{2}$ and $\mathcal{B}_w:=\mathcal{B}$ we get:
$$(\mathcal{A}-\mathcal{N})^*\mathcal{P}+\mathcal{P}(\mathcal{A}-\mathcal{N})-(1-\frac{1}{\gamma^2})\mathcal{P}\mathcal{B}\mathcal{B}^*\mathcal{P}+\mathcal{Q}=0.$$

Given that the obtained optimal value for $\gamma$ is approximately $\gamma \approx 574$, we can observe, based on the continuity of the Riccati solution concerning its entries, that the $H_\infty$ solution closely aligns with the LQR solution, especially since we have chosen $\mathcal{R}:=\mathcal{I}$.

\begin{figure}[h]
	\begin{center}
		\includegraphics[width=\linewidth]{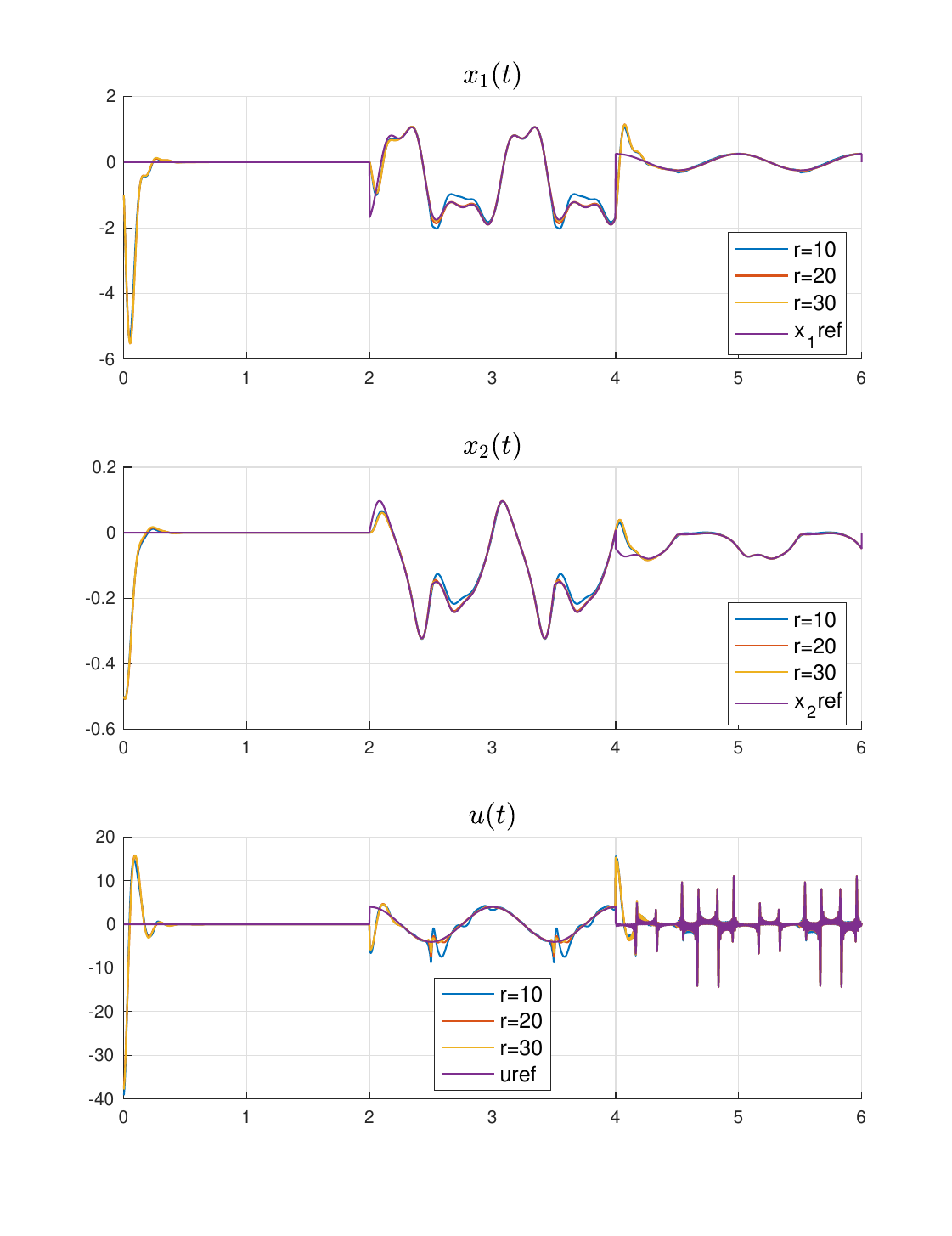}
		\caption{Closed loop response with optimal harmonic $H_\infty$ control for $r=10,20,30$}\label{f8}
	\end{center}
\end{figure}

\section{Conclusion}

In this paper, we have introduced a TBLMI-based framework for harmonic robust control design. This framework enables the formulation of robust control problems as semidefinite optimization problems, akin to those encountered in classical LTI systems. Notably, our approach deals with systems of infinite dimension. Leveraging the inherent Toeplitz structure, we have established a systematic and consistent methodology for solving these semidefinite optimization problems with precision, achieving solutions that are arbitrarily close to the ideal outcome. Our approach hinges on the definition of well-defined finite-dimensional truncated problems as the foundation for this endeavor. We provide illustrations addressing both $H_2$ and $H_\infty$ harmonic robust control problems. 

\bibliographystyle{ieeetr}
\bibliography{references}

\begin{thebibliography}{10}

\bibitem{Hara88}
S.~Hara, Y.~Yamamoto, T.~Omata, and M.~Nakano, ``Repetitive control system: a
  new type servo system for periodic exogenous signals,'' {\em IEEE
  Transactions on Automatic Control}, vol.~33, no.~7, pp.~659--668, 1988.

\bibitem{Lee98}
R.~C.~H. Lee and M.~C. Smith, ``Robustness and trade-offs in repetitive
  control,'' {\em Automatica}, vol.~34, pp.~889--896, 1998.

\bibitem{WEISS99}
G.~Weiss and M.~H{\"a}fele, ``Repetitive control of mimo systems using
  h$_\infty$ design,'' {\em Automatica}, vol.~35, no.~7, pp.~1185--1199, 1999.

\bibitem{Astolfi15}
G.~Weiss and M.~H{\"a}fele, ``Approximate regulation for nonlinear systems in
  presence of periodic disturbances,'' {\em In 54th IEEE conference on decision
  and control}, pp.~7665--7670, 2015.

\bibitem{Astolfiscl22}
D.~Astolfi, L.~Praly, and L.~Marconi, ``Harmonic internal models for
  structurally robust periodic output regulation,'' {\em Systems \& Control
  Letters}, vol.~161, p.~105154, 2022.

\bibitem{sanders1991generalized}
S.~R. Sanders, J.~M. Noworolski, X.~Z. Liu, and G.~C. Verghese, ``Generalized
  averaging method for power conversion circuits,'' {\em IEEE Transactions on
  Power Electronics}, vol.~6, no.~2, pp.~251--259, 1991.

\bibitem{wereley1990analysis}
N.~M. Wereley, {\em Analysis and control of linear periodically time varying
  systems}.
\newblock PhD thesis, Massachusetts Institute of Technology, 1990.

\bibitem{zhou_stability_2001}
J.~Zhou and T.~Hagiwara, ``Stability analysis of continuous-time periodic
  systems via the harmonic analysis,'' in {\em Proceedings of the 2001
  {American} {Control} {Conference}}, pp.~535--540 vol.1, 2001.

\bibitem{zhou2004spectral}
J.~Zhou, T.~Hagiwara, and M.~Araki, ``Spectral characteristics and eigenvalues
  computation of the harmonic state operators in continuous-time periodic
  systems,'' {\em Systems \& Control Letters}, vol.~53, no.~2, pp.~141--155,
  2004.

\bibitem{riedinger2022harmonic}
P.~Riedinger and J.~Daafouz, ``Harmonic pole placement,'' in {\em 2022 IEEE
  61st Conference on Decision and Control (CDC)}, pp.~5505--5510, 2022.

\bibitem{riedinger2022solving}
P.~Riedinger and J.~Daafouz, ``Solving {Infinite}-{Dimensional} {Harmonic}
  {Lyapunov} and {Riccati} {equations},'' {\em IEEE Transactions on Automatic
  Control}, vol.~68, no.~10, pp.~5938--5953, 2023.

\bibitem{zhou2002h2}
J.~Zhou and T.~Hagiwara, ``{$H_2$} and {$H_{\infty}$} norm computations of
  linear continuous-time periodic systems via the skew analysis of frequency
  response operators,'' {\em Automatica}, vol.~38, no.~8, pp.~1381--1387, 2002.

\bibitem{colaneri2000continuous}
P.~Colaneri, ``Continuous-time periodic systems in {$H_2$} and {$H_ {\infty}$}.
  {Part} {I}: Theoretical aspects,'' {\em Kybernetika}, vol.~36, no.~2,
  pp.~211--242, 2000.

\bibitem{colaneri2000continuous2}
P.~Colaneri, ``Continuous-time periodic systems in {$H_2$} and {$H_ {\infty}$}.
  {Part} {II}: State feedback problems,'' {\em Kybernetika}, vol.~36, no.~3,
  pp.~329--350, 2000.

\bibitem{Gero23}
J.~C. Geromel, ``Differential linear matrix inequalities in sampled-data
  systems filtering and control,'' {\em Springer}, 2023.

\bibitem{bamieh1991lifting}
B.~Bamieh, J.~B. Pearson, B.~A. Francis, and A.~Tannenbaum, ``A lifting
  technique for linear periodic systems with applications to sampled-data
  control,'' {\em Systems \& Control Letters}, vol.~17, no.~2, pp.~79--88,
  1991.

\bibitem{bamieh1992general}
B.~A. Bamieh and J.~B. Pearson, ``A general framework for linear periodic
  systems with applications to {$H_{\infty}$}/sampled-data control,'' {\em IEEE
  Transactions on Automatic Control}, vol.~37, no.~4, pp.~418--435, 1992.

\bibitem{yamamoto1996frequency}
Y.~Yamamoto and P.~P. Khargonekar, ``Frequency response of sampled-data
  systems,'' {\em IEEE Transactions on Automatic Control}, vol.~41, no.~2,
  pp.~166--176, 1996.

\bibitem{khargonekar1991h2}
P.~P. Khargonekar and N.~Sivashankar, ``H2 optimal control for sampled-data
  systems,'' {\em Systems \& Control Letters}, vol.~17, no.~6, pp.~425--436,
  1991.

\bibitem{zhou2008derivation}
J.~Zhou, ``Derivation and solution of harmonic {Riccati} equations via
  contraction mapping theorem,'' {\em Transactions of the Society of Instrument
  and Control Engineers}, vol.~44, no.~2, pp.~156--163, 2008.

\bibitem{floracdc23}
F.~Vernerey, P.~Riedinger, and J.~Daafouz, ``On solving infinite-dimensional
  toeplitz block lmis,'' {\em In 62nd IEEE conference on decision and control},
  2023.

\bibitem{blin_necessary_2022}
N.~Blin, P.~Riedinger, J.~Daafouz, L.~Grimaud, and P.~Feyel, ``Necessary and
  {Sufficient} {Conditions} for {Harmonic} {Control} in {Continuous} {Time},''
  {\em IEEE Transactions on Automatic Control}, vol.~67, no.~8, pp.~4013--4028,
  2022.

\bibitem{gohberg1993classes}
I.~Gohberg, S.~Goldberg, and M.~A. Kaashoek, ``Classes of linear operators.
  vol. {II},'' {\em Operator Theory: Advances and Applications}, vol.~63, 1993.

\bibitem{beam1993asymptotic}
R.~M. Beam and R.~F. Warming, ``The asymptotic spectra of banded toeplitz and
  quasi-toeplitz matrices,'' {\em SIAM Journal on Scientific Computing},
  vol.~14, no.~4, pp.~971--1006, 1993.

\bibitem{willems1971least}
J.~Willems, ``Least squares stationary optimal control and the algebraic
  {Riccati} equation,'' {\em IEEE Transactions on Automatic Control}, vol.~16,
  no.~6, pp.~621--634, 1971.

\bibitem{boyd1994linear}
S.~Boyd, L.~El~Ghaoui, E.~Feron, and V.~Balakrishnan, {\em Linear {matrix}
  {inequalities} in system and control theory}.
\newblock SIAM, 1994.

\bibitem{bensoussan1993standard}
A.~Bensoussan and P.~Bernhard, ``On the standard problem of
  {$H_{\infty}$}-optimal control for infinite-dimensional systems,'' {\em
  Identification and Control in Systems Governed by Partial Differential
  Equations}, vol.~68, p.~117, 1993.

\end{thebibliography}

\begin{IEEEbiography}[{\includegraphics[width=1in,height=1.25in,clip,keepaspectratio]{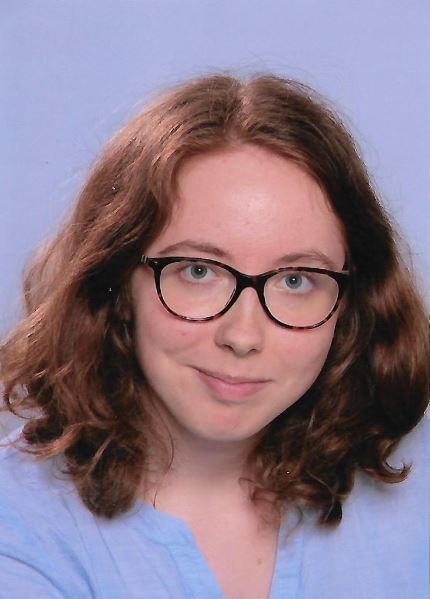}}]
	{Flora Vernerey} is a PhD student at CRAN - CNRS, Universit\'e de Lorraine (France).
 She received the engineering diploma (masters degree) in Digital Systems from ENSEM, Nancy, France, and the M.Sc. degree in Applied Mathematics from Université de Lorraine, Nancy, France, both in 2022. Her current research interests include identification of linear periodic systems, harmonic control and data-driven control.
\end{IEEEbiography}
\begin{IEEEbiography}[{\includegraphics[width=1in,height=1.25in,clip,keepaspectratio]{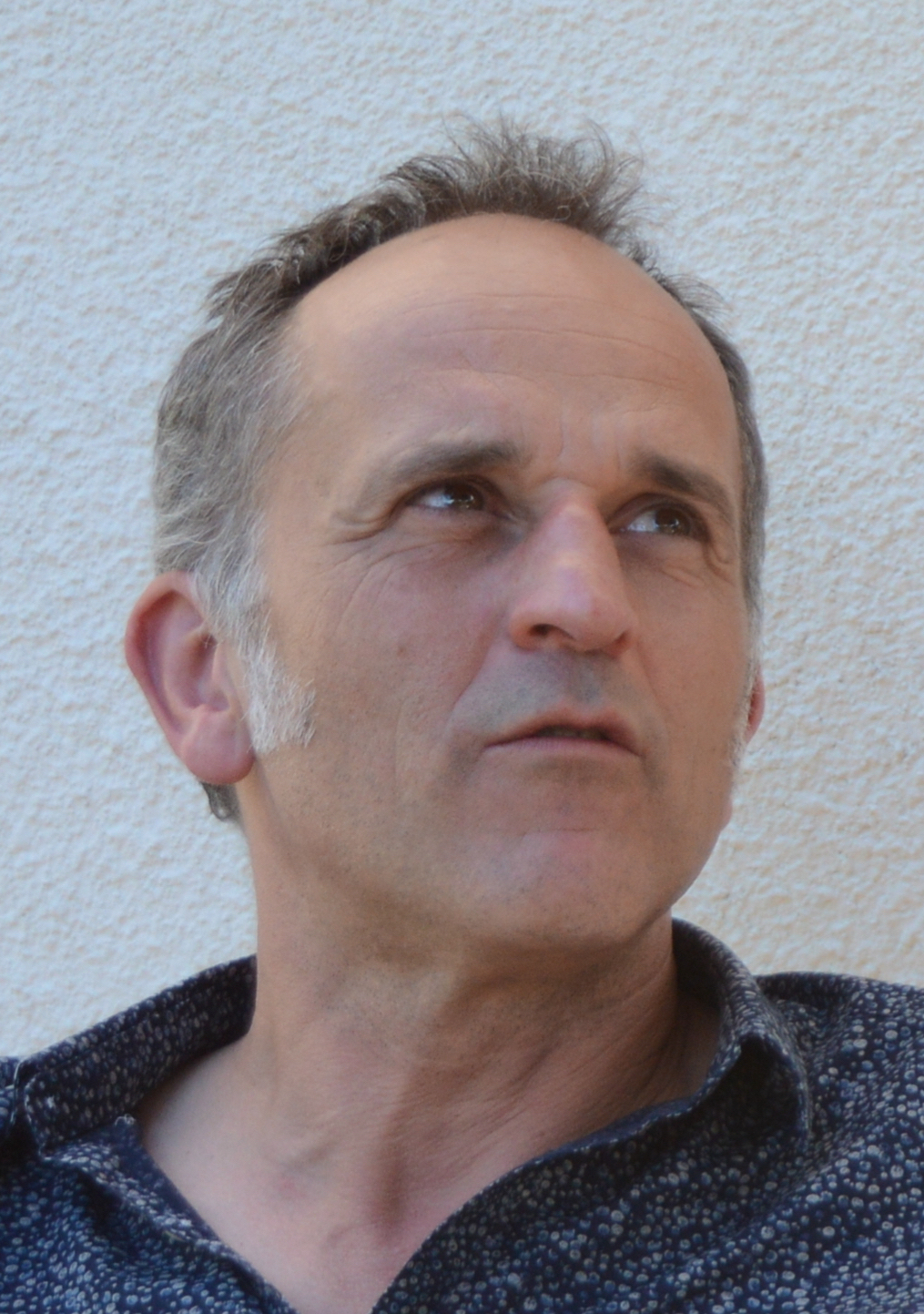}}]
	{Pierre Riedinger} is a Full Professor at the engineering school Ensem and researcher at CRAN - CNRS, Universit\'e de Lorraine (France).
	He received his M.Sc. degree in Applied Mathematics from the University Joseph Fourier, Grenoble in 1993 and the Ph.D. degree in Automatic Control 
	in 1999 from the Institut National Polytechnique de Lorraine (INPL). He got the French Habilitation degree from the INPL in 2010. His current research interests include control theory and optimization of 
	systems with their applications in electrical and power systems.
\end{IEEEbiography}
\begin{IEEEbiography}[{\includegraphics[width=1in,height=1.25in,clip,keepaspectratio]{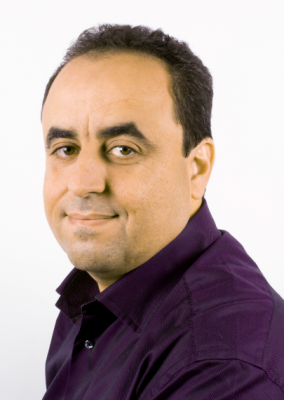}}]
	{Jamal Daafouz} 
	is a Full Professor at University
	de Lorraine (France) and researcher at CRAN-CNRS. In 1994, he received a Ph.D.
	in Automatic Control from INSA Toulouse, in 1997.
	He also received the "Habilitation \`a Diriger des
	Recherches" from INPL (University de Lorraine),
	Nancy, in 2005.
	His research interests include analysis, observation
	and control of uncertain systems, switched
	systems, hybrid systems, delay and networked systems with a particular
	interest for convex based optimisation methods.
	In 2010, Jamal Daafouz was appointed as a junior member of the
	Institut Universitaire de France (IUF). He served as an associate editor
	of the following journals: Automatica, IEEE Transactions on Automatic
	Control, European Journal of Control and Non linear Analysis and Hybrid
	Systems. He is senior editor of the journal IEEE Control Systems Letters. 
\end{IEEEbiography}
\end{document}